\documentclass[final]{dmtcs-episciences}

\usepackage[utf8]{inputenc}
\usepackage{caption}
\usepackage{subcaption}
\usepackage[round]{natbib}
\usepackage{amsmath,url}
\usepackage{amsfonts}
\usepackage{amssymb}
\usepackage{verbatim}
\usepackage{graphicx} 
\usepackage{color,multirow,multicol}

\newcommand*{\Union}{\bigcup}
\newtheorem{theorem}{Theorem}
\newtheorem{lemma}[theorem]{Lemma}
\newtheorem{corollary}[theorem]{Corollary}
\newtheorem{remark}[theorem]{Remark}
\newtheorem{proposition}[theorem]{Proposition}
\newtheorem{fact}{Fact}

\newcommand\mar[1]{\textcolor{black}{#1}}
\newcommand\ces[1]{\textcolor{black}{#1}}
\newcommand\fer[1]{\textcolor{black}{#1}}

\author{Ferhat Alkan\affiliationmark{1}
\and T\"{u}rker B{\i}y{\i}ko\u{g}lu\affiliationmark{2}
\thanks{Author is also partially supported by TUBA GEBiP/2009 and ESF EUROCORES TUBITAK Grant 210T173.}
\and Marc Demange\affiliationmark{3}
\and Cesim Erten\affiliationmark{4}
\thanks{Corresponding author. Part of this work was done while the author was visiting CIPF, Valencia. The author is also partially supported by TUBITAK-BIDEB Grant 1059B191501053.}}

\title{Structure of conflict graphs in constrained alignment  problems and algorithms\thanks{This work was partially supported by TUBITAK grant 112E137.}}

\affiliation{Division of Oncogenomics, The Netherlands Cancer Institute, Amsterdam, The Netherlands\\
2. Cadde, 12/9, 06500, Ankara, Turkey\\
School of  Science, RMIT University,
Melbourne, Australia\\
Computer Engineering, Antalya Bilim University, Antalya, Turkey}

\received{2017-8-14}
\revised{2019-7-13}
\accepted{2019-8-6}
\begin{document}
\publicationdetails{21}{2019}{4}{10}{3857}


\maketitle

 \begin{abstract}
We consider the constrained graph alignment problem which has applications in biological network analysis. 
Given two input graphs $G_1=(V_1,E_1), G_2=(V_2,E_2)$, two vertices \mar{$u_1,v_1$ of $G_1$ paired respectively to two vertices $u_2,v_2$ of $G_2$}  
induce an  \emph{ edge conservation} if 
\mar{$u_1,v_1$ and $u_2,v_2$ are adjacent in their respective graphs}. 
The goal is to provide a one-to-one mapping between \mar{some} vertices 
of the input graphs in order to maximize edge conservation. However the allowed mappings are restricted since each vertex from $V_1$ (resp. $V_2$) is allowed to be mapped to at most $m_1$ (resp. $m_2$) specified vertices in $V_2$ (resp. $V_1$). Most of the results in this paper deal with the case $m_2=1$ which attracted most attention in the related literature.
We formulate the problem as a maximum independent set problem in a related   {\em conflict graph} and investigate structural properties of this graph in terms of forbidden subgraphs. We are interested, in particular, in excluding certain wheels, fans, cliques or claws (all terms are defined in the paper), which in turn corresponds to excluding certain cycles, paths, cliques or independent sets in the neighborhood of each vertex. Then, we investigate algorithmic consequences  of some of these properties, which illustrates the potential of this approach and raises new horizons for further works. In particular this approach allows us to reinterpret a known polynomial case in terms of conflict graph and to improve known approximation and fixed-parameter tractability results through efficiently solving the maximum independent set problem in conflict graphs. Some of our new approximation results  involve  approximation ratios that are functions of the optimal value, in particular its square root; this  kind of results cannot be achieved for maximum independent set in general graphs.  
\end{abstract}

\section{Introduction}\label{sec:intro}
The  \emph{ graph alignment} problem has important applications in biological network alignment, in particular 
in the alignments of protein-protein interaction (PPI) 
networks~(\citet{AbakaBE13,AladagE13,sharan06,ZaslavskiyBV09,beams13}). 
Undirected graphs $G_1=(V_1,E_1), G_2=(V_2,E_2)$ (not necessarily connected) correspond to PPI networks 
for a pair of species, where  
the vertex sets $V_1, V_2$ represent the sets of proteins, and 
$E_1, E_2$ represent 
 the sets of known protein interactions pertaining to the   
networks of species under consideration. 
The informal goal is to find \mar{similar patterns between two PPI networks by identifying} a one-to-one mapping 
\mar{between some vertices of} $V_1$ and $V_2$ that maximizes the "similarity" of the mapped proteins, usually
scored with respect to the aminoacid sequence similarity 
and 
the conservation of interactions between mapped proteins. 
Functional orthology is an important application that
serves as the main motivation to study the alignment problems as part of a comparative analysis of PPI networks.  
\fer{A successful protein interaction network alignment across multiple species could provide a basis for deciding the proteins with similar functions, which may further 
be used in predicting functions of proteins with unknown functions or in verifying those with known functions, in detecting common orthologous pathways between species, or in reconstructing the evolutionary dynamics~(\citet{pmid28194172})}. 

A graph theory problem related to the biological network alignment problem is that 
of finding the  \emph{ maximum common edge subgraph} (MCES) of a pair of graphs, a problem 
commonly employed in 
the matchings of 2D/3D chemical structures~(\citet{Raymond02maximumcommon}). 
The MCES of two undirected graphs 
$G_1, G_2$ is a common subgraph (not necessarily induced) that contains the largest number of edges common to both $G_1$
and $G_2$. The NP-hardness of the MCES problem proposed in~\citet{GareyJ79} trivially implies 
that the biological network alignment problem is also NP-hard. 

A specific version of the problem
reduces its size by restricting the output alignment mappings to those 
chosen among certain subsets of protein mappings. The subsets of allowed mappings are assumed 
to be predetermined via some measure of similarity, usually that of sequence similarity~(\citet{AbakaBE13,ZaslavskiyBV09}). 
The  \emph{ constrained alignment} problem we consider herein can be considered as a graph theoretical generalization 
of this biological network alignment problem version. 
Formally, an instance $\prec G_1,G_2,S \succ$ is defined by  a pair of undirected graphs $G_1=(V_1,E_1), G_2=(V_2,E_2)$ and a bipartite graph $S=(V_1\cup V_2,E_S)$ with parts $V_1$ and $V_2$ \mar{representing possible matching between vertices of $G_1$ and vertices of $G_2$}. 
For $i=1,2$, we denote by $m_i$, the maximum degree in $S$ of vertices from part $V_i$.  
A  \emph{ legal alignment} $A$ is a matching of $S$, \mar{i.e., a set of independent edges (pairwise non adjacent). 
An edge $ab\in E_1$  is said to be {\em conserved}, if there is an edge $cd\in E_2$ such that $bc$ and $ad$ are in $A$, or $ac$ and $bd$ are in $A$. Then, the edge $cd$ is equivalently called conserved and, by definition of a matching, the number of conserved edges of $G_1$ is equal to the number of conserved edges of $G_2$}.
The constrained alignment problem is that of finding a legal alignment that 
maximizes \mar{the number of conserved edges in $G_1$ (or equivalently in $G_2$)}. 

Several related problems have been studied previously like, 
for instance, the {\em contact map overlap} problem introduced in~\citet{Goldman1999}. The goal  is to maximize the number of conserved edges;
however 
contrary to the constrained alignment problem, no constraint is given in terms of the bipartite graph $S$. Furthermore their problem definition assumes a linear order of the vertices of both $G_1, G_2$
which should be preserved by the output mapping. 
The problem 
of $(\mu_{G_1},\mu_{G_2})$-\emph{ matching with orthologies}, 
was introduced in~\citet{Fagnot2008}. Similar to the constrained alignment problem, it is to find a mapping 
\mar{respecting a set of constraints} represented by a bipartite graph $S$ \mar{but all edges of $G_1$ are requested to be conserved. 
Assuming $m_i=\mu_{G_i}, i=1,2$ and denoting by  
$\Delta_i=\Delta(G_i), i=1,2$  for an instance of the problem, where $\Delta(G)$ denotes the maximum degree of graph $G$, the problem of $(\mu_{G_1},\mu_{G_2})$-{matching with orthologies} is shown NP-complete 
even when $m_1=3, m_2=2$ and $G_1$ and $G_2$ are bipartite, $\Delta_1\leq 1$ and $\Delta_2\leq 2$, or if $m_1=3, m_2=1$ and $\Delta_1\leq 3, \Delta_2\leq 4$.} It is linear-time solvable if $m_1=2$ and  $m_2\mar{\in}O(1)$ (see also~\citet{Fertin200990}). \mar{Finally,  
 the problem  \emph{ MAX}$(\mu_{G_1},\mu_{G_2})$ 
 considered in~\citet{Fertin200990} is the optimization version of $(\mu_{G_1},\mu_{G_2})$-{matching with orthologies} with the objective to maximize the number of conserved edges. It is almost the same as the constrained alignment problem with $m_i=\mu_{G_i}, i=1,2$ with the additional requirement that every vertex of $G_1$ is  mapped to a vertex in $G_2$. We discuss more precisely the relations between these problems in Section~\ref{sec:definitions}}.  
 In~\citet{Fertin200990},   
 only the case $m_2=\mu_{G_2}=1$ is considered.  
It is shown APX-hard even if $m_1=2$ and $m_2=1$ (APX-complete if $G_1$  has bounded degree) and both graphs are bipartite.  
 They also propose several approximability and 
fixed-parameter tractability results (see~\citet{ausiellobook} and \citet{ParameterizedComplexity} for definitions about approximation and parameterized complexity, respectively). In particular,
they show that the problem can be approximated within ratio $2\lceil3\Delta_1/5\rceil$
for even $\Delta_1$ and ratio $2\lceil(3\Delta_1+2)/5\rceil$ for odd 
$\Delta_1$. They also show that the problem is fixed-parameter tractable 
on the size of the output assuming $m_2=1$, $m_1$ is constant and $G_1$ has a bounded degree.

In this paper, we consider the maximum constrained alignment problem as a maximum independent set  problem in a related  \emph{ conflict graph}, constructed from $G_1, G_2$, and $S$. Our aim is to investigate structural properties of this conflict graph in order to derive efficient  algorithms for the alignment problem. Although  
a conflict graph is also proposed in~\citet{Fertin200990} for $m_2=1$, 
with in particular a fixed-parameter tractability result based on a degree argument, 
no further structural property is provided. Here, we deepen this approach and strengthen algorithmic results.
Our main results and comparison with known results are given in Tables~\ref{table-structure},  \ref{table-approx} and \ref{table-parameter} 
\mar{at the end of this section}. 

Table~\ref{table-structure} shows our main structural results: the basic metrics of the graph \-- size and maximum degree \-- in the most general case as well as forbidden subgraphs for the case $m_2=1$. Some of these results have direct algorithmic consequences but even those without algorithmic applications are interesting, in particular since they motivate some graph classes for further studies. This is in particular the case for classes of graphs excluding some wheels or fans (related definitions are given in Section~\ref{sec:definitions}).

Table~\ref{table-approx} describes our approximation results that
extend the results in~\citet{Fertin200990} in several ways; it also illustrates  the potential of our approach. For instance,  an analysis of the degree of the conflict graph, generalizing the one in~\citet{Fertin200990},  immediately leads to an approximation ratio for the  general case  with a ratio $o(\Delta_1+\Delta_2)$ when $m_1, m_2$ are constant; it is improved to $o(\Delta_1)$ if $m_2=1$ and $m_1$ is constant. For the case $m_2=1$ and $m_1$ constant, we propose as well a $O\left(\frac{|V_1|}{\log(|V_1|)}\right)$-approximation as well as a $O(\sqrt{\beta(I)})$-approximation, where $\beta(I)$ is the optimal value of instance $I$. To our knowledge such kinds of ratios are totally new for this problem. Finally, one of our structural results gives a $(min(\Delta_1, \Delta_2)+1)$ approximation if $m_2=1$, improving also the previous known ratios. 

Table~\ref{table-parameter} presents two fixed parameter tractability results with respect to the size of the output. Both extend the results of~\citet{Fertin200990} to more general cases and both are direct consequences of structural results and known maximum independent set results. 

Finally, a last illustration of the potential of the maximum independent set approach is the case where $m_2=1$ and $G_1$ is acyclic. This case was already shown polynomial in~\citet{AbakaBE13}, using a specific dynamic programming method. A structural analysis of the conflict graphs allows to prove the same result and to interpret it as a maximum stable set polynomial case. Moreover it allows us to derive an explicit expression of the related complexity. Table~\ref{table-complex} sums-up all known complexity results for the maximum constrained alignment problem. \mar{Despite being obtained for {MAX}$(\mu_{G_1},\mu_{G_2})$ the hardness results also apply to the constrained alignment problem as noticed at the end of Section~\ref{sec:definitions}}.

 The paper is organized as follows. Section~\ref{sec:definitions} gives the main definitions, introduces the conflict graph and investigates its first characteristics (size and degree), leading  to first approximation and fixed parameter tractability results.
 Section~\ref{sec:m2is1} is dedicated to the case $m_2=1$ that raised the main attention in the literature.
 We first investigate in Subsection~\ref{subsec:struct-approx} some structural properties   of the conflict graph in terms of forbidden subgraphs (wheels and fans and cliques and claws) with their algorithmic consequences. This part constitutes our main contribution. Then, in Subsection~\ref{subsec:acyclic}, we revisit the case where $m_2=1$ and $G_1$ is acyclic. Finally Section~\ref{sec:conclude}  discusses further research directions.
 

\begin{center}
\begin{table}[ht]
\centering
\footnotesize
\begin{tabular}{|c|c|c|c|c|}

\hline 
$m_2$ &\multicolumn{1}{|c|}{$\geq 2$}& \multicolumn{3}{|c|}{$m_2=1$}\\
\hline
$m_1$ &\multicolumn{3}{|c|}{$\geq 3$}& $m_1=2$\\
\hline
$G_1$ and $G_2$& &$G_1$ acyclic&&\\
\hline
\hline
\multirow{4}{*}{Structural 
}&\multicolumn{1}{|c|}{$|V_{\mathcal{C}}|\leq \min\limits_{i=1,2}(m_i^2|E_i|)$}&&$W_t$-free, $t\geq 7$&$W_t$-free, $t\geq 5$\\
\multirow{4}{*}{property}&\multicolumn{1}{|c|}{(Lem.~\ref{vertexlemma})}&Weakly triangulated&\multicolumn{2}{|c|}{(Th.~\ref{lemW8})}\\
\cline{4-5}
&\multicolumn{1}{|c|}{}&(Th.~\ref{th:weaklytriangl})&$F_8$-free&$F_6$-free\\
\multirow{3}{*}{of $\mathcal{C}$}&\multicolumn{1}{|c|}{$\Delta(\mathcal{C})\leq \sum\limits_{i=1,2}2\Delta_im_i(m_i-1)$}&&\multicolumn{2}{|c|}{(Th.~\ref{Th:F8})}\\
\cline{4-5}
&\multicolumn{1}{|c|}{(Lem.~\ref{deglemma})}& &\multicolumn{2}{|c|}{$K_{1+{m_1}^2}$-free}\\
&\multicolumn{1}{|c|}{} &&\multicolumn{2}{|c|}{(Th.~\ref{lem_clique})}\\
\cline{4-5}
&\multicolumn{1}{|c|}{Bound of $|E_{\mathcal{C}}|$ using the first Zagreb Index}&&\multicolumn{2}{|c|}{$(2\Delta_{min}+2)$-free}\\
&\multicolumn{1}{|c|}{(Lem.~\ref{edgelemma})}&&\multicolumn{2}{|c|}{(Th.~\ref{claw})}\\
\hline
\end{tabular}
\caption{Main structural Properties of $\mathcal{C}$.}
\label{table-structure}
\end{table}
\end{center}

\begin{center}
\begin{table}[ht]
\centering
\footnotesize
\begin{tabular}{|c|c|c|c|}
\hline 
$m_2$ &\multicolumn{1}{|c|}{$\geq 2$}& \multicolumn{2}{|c|}{$m_2=1$}\\
\hline
$m_1$ &\multicolumn{2}{|c|}{$\geq 3$}& $m_1=2$\\
\hline
\hline
\multirow{12}{*}{Approximation ratio} &\multirow{6}{*}{$O\left(\frac{(\Delta_1+\Delta_2)\log \log(\Delta_1+\Delta_2)}{\log(\Delta_1+\Delta_2)}\right)$}&\multicolumn{2}{|c|}{$\frac{6\Delta_1}{5}+{\rm cst}$}\\
&&\multicolumn{2}{|c|}{(\citet{Fertin200990})}\\
\cline{3-4}
&&\multicolumn{2}{|c|}{$O\left(\frac{\Delta_1\log \log(\Delta_1)}{\log(\Delta_1)}\right)$}\\
&&\multicolumn{2}{|c|}{($m_1$ constant \-- Prop.~\ref{cor:approxD1})}\\
\cline{3-4}
&($m_i$ const., $i=1,2$ \-- Prop.~\ref{prop:generalapprox})&$\sqrt{3\beta(I)/2}$&$\sqrt{\beta(I)}$\\
&&\multicolumn{2}{|c|}{(Prop.~\ref{pro:approxsqrt})}\\
\cline{3-4}
&&\multicolumn{2}{|c|}{$\forall K>0, \left\lceil \frac{|V_1|}{K\log(|V_1|)}\right\rceil$}\\
&&\multicolumn{2}{|c|}{($m_1$ constant \-- Th.~\ref{th:approxlog})}\\
\cline{3-4}
&&\multicolumn{2}{|c|}{$\Delta_{min}+1$}\\
&&\multicolumn{2}{|c|}{(Prop.~\ref{cor16})}\\
\hline
\end{tabular}
\caption{Main approximation results ($\beta(I)$ denotes the optimal value of instance $I$).}
\label{table-approx}
\end{table}
\end{center}

\begin{center}
\begin{table}[ht]
\centering
\footnotesize
\begin{tabular}{|c|c|c|c|}
\hline 
$m_2$ &\multicolumn{1}{|c|}{Bounded $\geq 2$}& \multicolumn{2}{|c|}{$m_2=1$}\\
\hline
$m_1$ &\multicolumn{3}{|c|}{Bounded $\geq 3$}\\
\hline
$G_1$ and $G_2$& \multicolumn{2}{|c|}{Bounded degree}&Any degree\\
\hline
\hline
\multirow{2}{*}{Parameterized tractability} &FTP&FTP&FTP\\
&(Prop.~\ref{prop:parameter1})&(\citet{Fertin200990})&(Prop.~\ref{cor14})\\
\hline
\end{tabular}
\caption{FTP results parameterized by the size of the output}
\label{table-parameter}
\end{table}
\end{center}

\begin{center}
\begin{table}[ht]
\centering
\footnotesize
\begin{tabular}{|c|c|c|c|}

\hline 
$m_2$ &\multicolumn{3}{|c|}{$\geq 1$}\\
\hline
$m_1$ &\multicolumn{3}{|c|}{$\geq 2$}\\

\hline
\multirow{2}{*}{$G_1$ and $G_2$}&\multicolumn{2}{|c|}{ 
Even bipartite}&\multirow{2}{*}{$G_1$ acyclic}\\
\cline{2-3}
&Any degree & Bounded degree&\\
\hline
\hline
\multirow{2}{*}{Complexity}  &APX-hard&APX-complete&Polynomial\\
&\multicolumn{2}{|c|}{(\citet{Fertin200990})}&(\citet{AbakaBE13} and~Subs.~\ref{subsec:acyclic})\\
\hline
\end{tabular}
\caption{Complexity of the constrained alignment problem}
\label{table-complex}
\end{table}
\end{center}


\newpage

 \section{Definitions and first remarks}\label{sec:definitions}

\subsection{Main definitions and the considered problem}\label{sub:def}

For all graph-theoretical definitions not given here, the reader is referred to~\citet{golumbicbook}. \mar{A  \emph{ matching} in a graph is a set of independent edges, i.e., pairwise non adjacent. The extremities of the edges in the matching are called  \emph{ saturated}.}
For any $t\geq 2$, $P_t$ denotes a path with $t$ vertices ( \emph{ t-path}), $C_t$ denotes a  cycle with $t$ vertices ( \emph{ t-cycle}) and $K_t$ denotes a clique with $t$ vertices ( \emph{ t-clique}). \mar{A $P_t$ or a $C_t$ will be denoted as list of successive vertices like $x_1x_2\cdots x_t$. In the case of a $t$-path $x_1$ and $x_t$ are the extremities while, in the case of a $t$-cycle, $x_1$ is any vertex and the order correspond to one of the two possible orientations of the cycle. Sometimes, when a confusion is possible, the $t$-cycle will be denoted $x_1x_2\cdots x_tx_1$ to distinguish it from a $t$-path.}
Denote the complement of $G$ with $\overline{G}$. An  \emph{ induced subgraph} of $G=(V,E)$ is a subgraph of $G$ induced by a subset of vertices, $X\subset V$. It will be denoted by $G[X]$. Given a graph $H$, $G$ will be called {\em $H$-free} if it does not have any induced subgraph isomorphic to $H$. A partial graph of $G=(V,E)$ is a graph $G'=(V,E')$ with $E'\subset E$ and a partial induced subgraph is a partial graph of an induced subgraph. For a vertex $v\in V$ we will denote by $N(v)$ its (open) neighborhood and by $N[v]=N(v)\cup\{v\}$ its close neighborhood. For any vertex $v$ we will denote by $G_v=G[N[v]]$ the subgraph induced by $v$ and its neighborhood. For a vertex $v\in V$, $d_G(v)$ is its degree in $G$. When no ambiguity may occur, we simply denote $\Delta$ instead of $\Delta(G)=\max_{v\in V}(d_G(v))$.

A graph is called {\em weakly triangulated} if it is $C_t$-free and $\overline{C_t}$-free, for $t\geq 5$.

For $t\geq 3$, a {\em wheel} $W_t$ is a graph consisting of a $t$-cycle $C_t$ with an additional vertex, called {\em center}, adjacent to all the vertices of the  cycle $C_t$. A fan graph $F_t$ consists of a path $P_t$ with $t$ vertices and a new vertex $v$ 
that is adjacent to all the vertices of the path. As a consequence, a graph $G=(V,E)$ is $W_t$-free (resp. $F_t$-free) if and only if, for every vertex $v\in V$, $G_v$ is $C_t$-free (resp. $P_t$-free).

An {\em independent set} is a set of  pairwise non adjacent vertices, i.e., they induced a graph without any edge. 
 Given a graph $G$, $\alpha(G)$ denotes its independent number, i.e., the maximum size of an independent set in $G$. Consider a graph class $\mathcal{G}$ and a polynomial algorithm 
 determining, for every graph $G\in \mathcal{G}$ of a graph class, an independent set of size $\lambda(G)$, is said to guarantee the approximation ratio of $\rho(G)$, for a function $\rho\geq 1$, on $\mathcal{G}$ if:
 $$\forall G\in \mathcal{G},\  \frac{\alpha(G)}{\lambda(G)}\leq \rho(G)$$
 
 Polynomial approximation algorithms are defined similarly for other graph maximization problems. \mar{If an algorithm guarantees a ratio that belongs to the class of functions $O(f)$ (resp. $o(f)$), then we will simply say that the algorithm guarantees a ratio of $O(f)$ (resp. $o(f)$) or constitutes a $O(f)$- (resp. $o(f)$-)approximation.} The reader is referred to \cite{ausiellobook} for all concepts in approximation not defined here. Throughout the paper we only use natural logarithms, so $\log$ stands for $\log_e$.   

\mar{Finally, in Subsection~\ref{sec:anym1m2}, we will use the {\em first Zagreb index} of a graph $G$; it is denoted $M_1(G)$. $M_1(G)$ is defined as the sum of  squares of  degrees of the vertices. It has been extensively studied, in particular for its interest in computational chemistry~(see, e.g. \citet{Zagreb-30years} for an introduction to this index).}\\

\noindent
\mar{The constrained alignment problem is formally defined as follows:}

\mar{
\begin{tabular}{ll}
\textbf{ Input:}& $I=\prec G_1, G_2, S\succ$, where $G_1=(V_1,E_1), G_2=(V_2,E_2)$ are undirected graphs \\
& and  $S=(V_1\cup V_2, E_S)$  is a bipartite graph with parts $V_1,V_2$;\\
& $I$ will be called an {\em instance}.\\
\textbf{ Output:}& A matching $A$ of $S$, called legal alignment;\\
\textbf{ Objective:}& Maximize the number of conserved edges in $G_1$, or equivalently in $G_2$, i.e., \\
& the number of pairs $(ab,cd)\in E_1\times E_2$, where $ad,bc\in A$ or $ac,bd\in A$.\\
&
\end{tabular}}

\ces{For the ease of description, the edges of the bipartite graph $S$ will be called  \emph{ similarity edges}.}
\mar{A legal alignment is called {\em minimal} if the removal of any similarity edge in the alignment creates an alignment that conserves less edges. Any legal alignment includes at least one minimal alignment and consequently,  an optimal minimal alignment is an optimal alignment. Therefore, we can restrict ourselves to minimal alignments.}

\mar{We conclude this subsection with few remarks comparing the constrained alignment problem and relaqted problems introduced in Section~\ref{sec:intro}. Note that the conserved edges of $G_1$ and $G_2$ as well as their extremities respectively induce  isomorphic partial subgraphs of $G_1$ and  $G_2$. So, if $S$ is a complete bipartite graph, then the problem corresponds to finding two isomorphic partial subgraphs of $G_1$ and $G_2$ with a maximum number of edges, which is exactly the maximum common edge subgraph.} 
\mar{However, in our case, the bipartite graph $S$ constraints the possible  isomorphisms since a vertex of $V_1$ (resp. $V_2$) can only be mapped to one of its neighbors in $S$. In an applied context, such constraints represent a priori knowledge about the system that makes only some matchings meaningful.}

\mar{The only difference with the problem  \emph{ MAX}$(\mu_{G_1},\mu_{G_2})$,  with $m_i=\mu_{G_i}, i=1,2$ 
 (\citet{Fertin200990}), is that in this latter problem, the matching $A$ is required to saturate all vertices in $G_1$, thus defining an injective (one-to-one) mapping from $V_1$ to $V_2$. Contrary to the problems considered in~\citet{Fagnot2008,Fertin200990}, our problem is symmetric in $G_1,G_2$. All our results can be equivalently formulated by swapping indexes 1 and 2. When we will assume that one of $G_1,G_2$ has a specific structure, in particular acyclic like in Subsection~\ref{subsec:acyclic}, we can assume without loss of generality that the condition holds for $G_1$.  Roughly speaking, the problems considered in~\citet{Fagnot2008,Fertin200990} correspond to detecting, in $G_2$ a specific structure as close as possible to the pattern represented by $G_1$. Our version however,   aims to detect similar patterns in the two graphs. We believe that both versions make sense for the suggested  applications.}

 \mar{With the constraint for the solution to define an injective mapping from $V_1$ to $V_2$, some instances of {MAX}$(\mu_{G_1},\mu_{G_2})$ may have no feasible solution while every instance of the constrained alignment problem has at least one feasible solution. For this reason, \citet{Fertin200990} restrict their problem to the so called  \emph{ trim instances} for which $S$ has a matching saturating $V_1$, every vertex in $V_2$ has a degree at least~1 in $S$ and there is no  \emph{ bad edge} in $G_1$, i.e., an edge that cannot be conserved for any matching of $S$. The constrained alignment problem does not require  the first assumption. Removing bad edges as well as isolated vertices in $S$ can be performed in polynomial time and leads to an equivalent instance. So, we can assume that there is neither bad edge nor isolated vertex in $S$.} 
 
 \mar{Note finally that any $(m_1,m_2)$-instance of the constrained alignment problem (with $m_2>0$) can be transformed into an  instance of MAX($m_1+1, m_2$) with the same optimal value by adding to $V_2$ a set $V_I$ of $|V_1|$ independent vertices and linking, in $S$, every vertex in $V_1$ to its copy in $V_I$. This transformation does not modify $m_2$. In addition, note that, with the restriction that $S$ has no isolated vertex, the alignment problem with $m_2=1$ is equivalent to {MAX}$(\mu_{G_1},1)$ problem and if there is no bad edge, then all instances are trim instances for the latter problem. Indeed, if all vertices of $V_1$  have a degree at least~1 in $S$ and if vertices in $V_2$ have the degree~1 $S$, then all maximal matchings of the graph $S$ saturate $V_1$. As a consequence, all known results  for {MAX}$(\mu_{G_1},1)$ also hold for the alignment problem with $m_2=1$.}

\subsection{Conflict graph}\label{subsec:conflict}
 
 \subsubsection{The notion of $c_4$s and their conflicting configurations}
\ces{For the following,  we will call $c_4$ some specific 4-cycles   $abcd$, where 
$ab\in E_1$, $cd \in E_2$ and $ad,bc \in E_S$. These are partial induced $C_4$'s of the graph $(V_1\cup V_2, E_1\cup E_2\cup E_S)$, obtained as the union of $G_1, G_2$ and $S$, 
for the instance $\prec G_1, G_2, S \succ$.} 
\mar{Throughout the paper, we adopt the following notations to avoid any confusion between the different graphs we will refer to. When referring to $c_4$s, we will use simple letters from $a$ to $w$ (without indexes) to denote vertices of $V_1\cup V_2$. A $c_4$ is then denoted as a list of four vertices, where the two first ones are in $V_1$ and the two last are in $V_2$. Letters $x,y,z$ (sometimes with indexes) will denote vertices of the conflict graph defined below.}

We say that
two $c_4$s  \emph{ conflict}, if at least two of their similarity edges 
are adjacent \mar{but distinct} (then, they cannot coexist in any matching of $S$).
Let $efgh$  be a $c_4$ conflicting with the $c_4$ $abcd$, where $ef\in E_1$, $gh\in E_2$, and $eh, fg\in E_S$. 
In the case $m_2=1$, we can identify five generic configurations corresponding to the relative position of $efgh$ 
with respect to $abcd$. These possible configurations are shown in Figure~\ref{twoc4}; note that if $e,f\in\{a,b\}$ or $g,h\in\{cd\}$, then only the label in $\{a,b,c,d\}$ is represented. In $Conf_{1a}$, we have $a=f$, and the rest of the vertices are all distinct; in this case, we say that it is a $Conf_{1a}$ conflict. Analogously, in 
$Conf_{1b}$, we have $b=e$, and the rest of the vertices are all distinct. In $Conf_{2}$, we have
$a=e, b=f$, and the rest of the vertices are all distinct. In $Conf_{3a}$, we have $a=e, b=f,c=g$, and 
the rest of the vertices $d,h$ are distinct. Analogously, in $Conf_{3b}$, we have $a=e, b=f, d=h$, and the rest of the vertices 
$c, g$ are distinct. So, the number in the name of the conflicting configuration represents the number of vertices the two $c_4$s have in common. Similar to $Conf_{1a}$ conflict, we will refer to $Conf_{1b}$, $Conf_{2}$, $Conf_{3a}$ or $Conf_{1b}$ conflicts. 

\begin{figure}[t]	   
\begin{center}	   
\includegraphics[width=12cm]{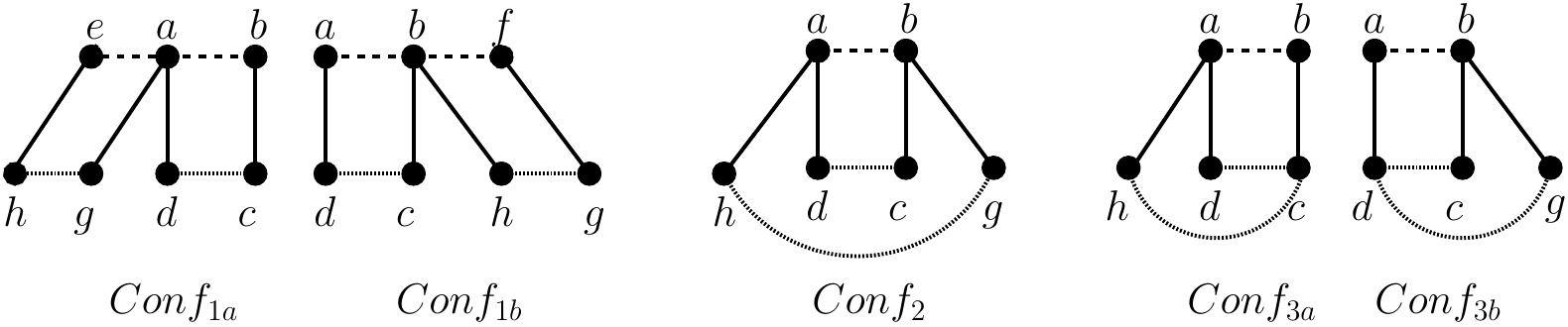} 
\caption{\sf \ces{Given two conflicting $c_4$s, $abcd$ and $efgh$, all possible conflicting configurations with respect to $abcd$, when $m_2=1$. 
For each configuration, the vertices at the top are $V_1$ vertices and the vertices at the bottom are $V_2$ vertices. 
}}
\label{twoc4}	   
\end{center}	   
\vspace*{-.4cm}	   
\end{figure}


For larger $m_2$, one can also observe all symmetric conflicting configurations obtained by exchanging $V_1$ and $V_2$ with similarity edges adjacent on $V_2$ vertices plus one configuration with two similarity edges adjacent on a $V_1$ vertex and two adjacent on a $V_2$ vertex.

\subsubsection{The conflict graph and its independent sets}

\mar{With} a given instance $\prec G_1,G_2,S \succ$, we 
\mar{associate} a  \emph{ conflict graph}, \mar{$\mathcal{C}=\left(V_{\mathcal{C}}, E_{\mathcal{C}}\right)$}, as follows. 
 For each $c_4$, create a vertex in \mar{$V_{\mathcal{C}}$} and for each pair of conflicting $c_4$s, create an edge between their respective vertices in \mar{$E_{\mathcal{C}}$}. 

\mar{We will denote by $\gamma$ the one-to-one correspondence mapping vertices of the conflict graph $\mathcal{C}$ to $c_4$s in 
$(V_1\cup V_2, E_1\cup E_2\cup E_S)$.
Thus, 
for any vertex $x\in V_{\mathcal{C}}$ of the conflict graph, 
$\gamma(x)$ is} the corresponding $c_4$; for instance, if the related $c_4$ is $abcd$ with $a,b \in V_1, ab\in E_1$ and $c,d\in V_2, cd\in E_2$, we will write $\gamma(x)=abcd$. We 
call $\gamma(x)$ ``the $c_4$ \mar{associated with} $x$''. \mar{In Theorem~\ref{Th:F8}, we will need the notation $\gamma(x)\cap \{a,b\}$ to denote the set of vertices in $\{a,b\}$ and visited by the $c_4$ $\gamma(x)$.}  



With this construction of the conflict graph, the constrained alignment problem  reduces to the maximum independent set problem \mar{as stated in the following proposition. This will be illustrated in the example detailed in Paragraph~\ref{subsub:figure}.}


\begin{proposition}\label{prop:reduction_alpha}\mbox{}\\
{\bf{(i)}} There is a one-to-one correspondence (bijective mapping) between independent sets in the conflict graph and minimal alignments in the instance $\prec G_1,G_2,S \succ$. An independent set of $p$ vertices maps to an alignment that conserves $p$ edges.\\
  {\bf{(ii)}} A maximum independent set of $\mathcal{C}$ maps to to an optimal alignment for $\prec G_1,G_2,S \succ$.\\
   {\bf{(iii)}} The maximum possible number of conserved edges is $\alpha(\mathcal{C})$.
\end{proposition}
\begin{proof}
\textbf{(i)}
 Let $\{x_1, \ldots, x_p\}$ be an independent  set in the conflict graph $\mathcal{C}$; by definition of the conflict graph, the $c_4$s $\gamma(x_i), i=1, \ldots, p$ are pairwise not conflicting in the graph $(V_1\cup V_2, E_1\cup E_2\cup E_S)$ and consequently their similarity edges constitute a legal alignment $A$. An edge $ab\in E_1$ is conserved for this alignment if and only if there are two edges $ad, bc$ in $A$ and $cd\in E_2$; in this case $abcd=\gamma(x_i)$ for some $i\in\{1, \ldots, p\}$.  Since two distinct non conflicting $c_4$s cannot share an edge of $G_1$ (neither of $G_2$), exactly $p$ edges of $G_1$ are conserved by this alignment. This also implies that the alignment $A$ is minimal.
 
 Conversely, for any minimal legal alignment that conserves $p$ edges of $G_1$, the conserved edges are in one-to-one correspondence with non-conflicting $c_4$s in the graph $(V_1\cup V_2, E_1\cup E_2\cup E_S)$. Through $\gamma^{-1}$, these $c_4$s correspond to an independent set $\{x_1, \ldots, x_p\}$ in $\mathcal{C}$.
 
 \textbf{(ii)} Since the one-to-one correspondence transforms an independent of cardinality $p$ set into  an alignment conserving $p$ edges, a maximum independent set maps to an alignment maximising the number of conserved edges.

 \textbf{(iii)} It follows immediately that the maximum possible number of conserved edges is $\alpha(\mathcal{C})$.
 
\end{proof}

\begin{corollary}\label{cor:reduction}
\mar{Any polynomial approximation algorithm for the maximum independent set in a graph $G$  guaranteeing the ratio $\rho(\mathcal{G})$  can be turned into a  polynomial approximation algorithm for the constrained alignment problem guaranteeing the ratio $\rho(\mathcal{\mathcal{C}})$, where $\mathcal{C}$ is the conflict graph associated with the instance $\prec G_1,G_2,S \succ$.}
\end{corollary}

\begin{proof}
\mar{The conflict graph as well as the mapping $\gamma$  can be computed in polynomial time with respect to the size $|V_1|+|V_2|$ of the instance  $\prec G_1,G_2,S \succ$ since it only requires identifying all $c_4$s and testing the compatibility of every two $c_4$s. The conflict graph is of polynomial size (details about its size are given in Subsection~\ref{sec:anym1m2}) and it follows immediately from the proof of Proposition~\ref{prop:reduction_alpha}-(i) that, given an independent set of size $p$ in $\mathcal{C}$, computing the corresponding minimal alignment that conserves $p$ edges can be done in polynomial. We conclude by using the fact that the maximum possible number of conserved edges is $\alpha(\mathcal{C})$.}
\end{proof}

\mar{Approximation ratios for the maximum independent set problem are usually expressed as functions of the number of vertices and/or maximum degree of the graph instance.  To derive an approximation ratio for the constrained alignment expressed as a function of the instance $\prec G_1,G_2,S \succ$ will require evaluating the main parameters of the conflict graph. This is the purpose of the Subsection~\ref{sec:anym1m2}.}

\mar{\begin{remark}\label{rem:noninjective}
Several minimal alignments (thus, several independent sets of the conflict graph) may correspond to the same set of conserved edges.
\end{remark}}

 \mar{Consider for instance as the graph $G_1$ a path $abc$ of length 2 and  as the graph $G_2$ a path $def$. If similarity edges are $ad, be, cf, af$  and $cd$, then, the two minimal alignments $\{ad, be, cf\}$ and $\{af, be, cd\}$ conserve the same edges $ab$ and $bc$ of $G_1$. We give in paragraph~\ref{subsub:figure} another possible situation, where two different alignments correspond to the same conserved edges in $G_1$ but not in $G_2$.}

\subsubsection{The underlying graph}

\mar{A direct consequence of Proposition~\ref{prop:reduction_alpha} is that removing from the instance  $\prec G_1,G_2,S \succ$ all $G_1$-edges, $G_2$-edges or similarity edges that do not belong to any $c_4$ does not change the problem in the sense that minimal alignments remain the same. For this reason, we  
 consider the graph $\mathcal{C}_U=(V_U,E_U)$ 
obtained from the union of $G_1, G_2$, and $S$ by excluding all the vertices and edges that are not part of 
any $c_4$s. In particular, this includes removing all bad edges~(\citet{Fertin200990}) of $G_1$ and $G_2$. We call $\mathcal{C}_U$ the {\em underlying graph} associated with the instance $\prec G_1,G_2,S \succ$.  It can be seen  as a simplified equivalent instance and consequently, we can always assume that we work on $\mathcal{C}_U$ instead of $(V_1\cup V_2, E_1\cup E_2\cup E_S)$ or, equivalently, that each edge in $E_1\cup E_2\cup E_S$ belongs to at least one $c_4$. In particular, in all our results, $m_i$ can be seen as the maximum number of similarity edges in $E_U$ incident to vertices of $V_i\cap V_U$.}   

\subsubsection{An example}\label{subsub:figure}

\mar{Figure~\ref{sample} gives an example that illustrates the notions of conflict graph, of underlying graph, the function $\gamma$ and the correspondence between minimal alignments in the original instance and independent sets in the conflict graph.  The left chart represents the instance $I=\prec G_1, G_2,S \succ$ and the related underlying graph $\mathcal{C}_U$. $G_1$ is represented on the top, with vertices $V_1=\{a,b,c,d,e\}$ and dashed edges and $G_2$ on the bottom  with vertices $V_2=\{f,g,h,i,j\}$ and dotted edges.  Blue edges/vertices correspond to edges/vertices in $(V_1\cup V_2, E_1\cup E_2\cup E_S)$ that are not part of the underlying graph. So, the underlying graph $\mathcal{C}_U$ appears in black color. In the original instance $m_1=3$ but, in the equivalent simplified instance defined by $\mathcal{C}_U$, it becomes 2.}

\mar{The list of $c_4$s and the related function $\gamma$
are represented in the middle part of the figure.
Note that $adcb$ or $bhcg$ are 4-cycles in $\mathcal{C}_U$ but not $c_4$s. }

\mar{Finally, the related conflict graph is represented on the right hand side. 
This instance has four different optimal solutions corresponding to the minimal alignments $\{ag,bh,di\}, \{bh,cg,di\}, \{bg,ch,di\}$ and $\{ag,di,ch\}$. They correspond respectively to the independent sets $\{x_1,x_4\}, \{x_2,x_6\}, \{x_3,x_5\}$ and $\{x_4,x_5\}$ in the conflict graph. Each optimal solution corresponds to two conserved edges in $E_1$: $\{ab,ad\}$, $\{bc,cd\}, \{bc,cd\}$ and $\{ad,cd\}$, respectively. In this example, these conserved edges correspond to an induced $P_3$ in the graph $G_1$ but, in the graph $G_2$, the related conserved edges which are respectively  $\{gh,gi\}, \{gh,hi\}, \{gh,gi\}$ and $\{gi,hi\}$, are not induced subgraphs of $G_2$ but only partial induced subgraphs. Note finally that the two alignments $\{bh,cg,di\}$ and $\{bg,ch,di\}$ correspond to the same conserved edges in $G_1$ but not in $G_2$. This is another illustration of Remark~\ref{rem:noninjective}.}

In what follows we provide several graph-theoretic properties of conflict graphs
arising from possible constrained alignment instances under various restrictions. 
Such properties are then employed in applying relevant independent set 
results. 

Throughout the paper we will assume $|V_1|\geq 2$ and $|V_2|\geq 2$ since, in the opposite case, the conflict graph is empty and the maximum alignment problem would be trivial \mar{(the only minimal alignment is empty)}. 
For a vertex $x\in V_i$ of $G_i$, $i=1,2$, we will denote by $d_i(x)$ its degree in $G_i$.

\begin{figure}[t]	   
\begin{center}	   
\includegraphics[width=12cm]{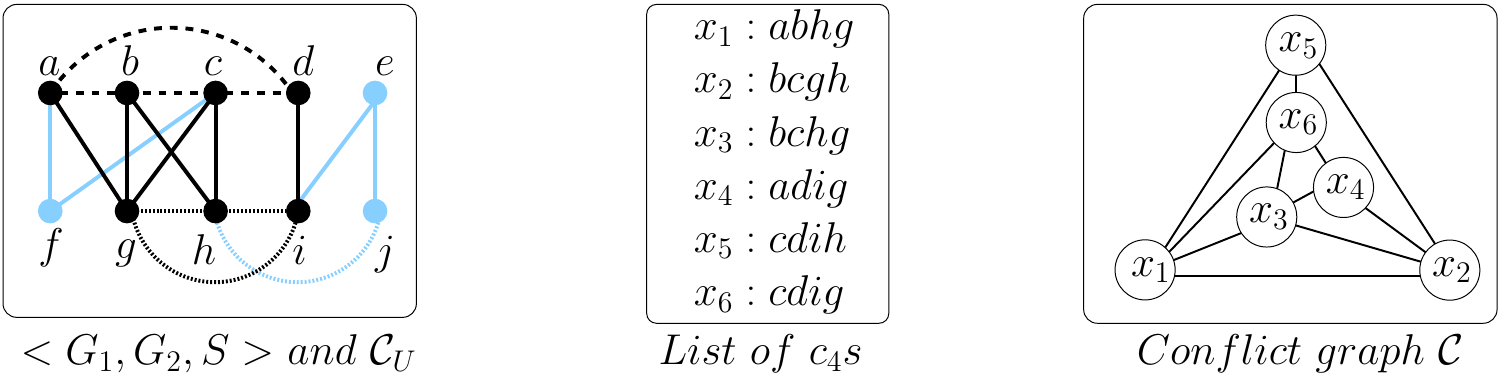} 
\caption{\sf \mar{An instance $\prec G_1,G_2,S\succ$ with the underlying graph $\mathcal{C}_U$ and the conflict graph $\mathcal{C}$. $V_1=\{a,b,c,d,e\}$ and $V_2=\{f,g,h,i,j\}$. In the left graph, dashed lines correspond to edges in $E_1$ while dotted lines correspond to edges in $E_2$. Blue edges and vertices (left graph) are edges and vertices in $(V_1\cup V_2, E_1\cup E_2\cup E_S)$ that are not part of the underlying graph and can be ignored. The list of $c_4$s also defines the function $\gamma$. }
} 
\label{sample}	   
\end{center}	   
\vspace*{-.4cm}	   
\end{figure}

\subsection{General properties of the conflict graph and applications}\label{sec:anym1m2}
In this subsection we first investigate the first basic properties of the conflict graph and deduce first approximation results using some standard results on the maximum independent set problem. For an instance $\prec G_1,G_2,S \succ$, we denote by $\mathcal{C}=(V_\mathcal{C},E_\mathcal{C})$ the related conflict graph.

\begin{lemma}
\label{vertexlemma}
Given an instance $\prec G_1,G_2,S \succ$ with conflict graph $\mathcal{C}$, the number $|V_\mathcal{C}|$ of vertices of $\mathcal{C}$ satisfies:\\ 
{\centerline{$|V_\mathcal{C}|\leq \min \left({m_1}^2|E_1|,{m_2}^2|E_2|, \frac{1}{2}m_1m_2|V_1|\Delta_2,\frac{1}{2}m_1m_2|V_2|\Delta_1\right)$.}}
\end{lemma}
\begin{proof}
Consider a similarity edge $xy\in E_S$, $x\in V_1, y\in V_2$.  
The edge $xy$ can belong to at most $\min(m_1d_1(x), m_2d_2(y))$ different $c_4$s. Consequently the number of possible  $c_4$s satisfies: 
$$|V_\mathcal{C}|\leq\frac{1}{2}\sum_{xy\in E_S}\min(m_1d_1(x), m_2d_2(y)).$$
Since $x$ has at most $m_1$ incident edges in $S$ and $d_2(y)\leq \Delta_2$ we deduce:  
$$|V_\mathcal{C}|\leq \frac{m_1}{2}\sum_{x\in V_1}\min(m_1d_1(x), m_2\Delta_2)\leq \min({m_1}^2|E_1|, \frac{1}{2}m_2m_1 |V_1|\Delta_2).$$ 

Similarly we have:
$$|V_\mathcal{C}|\leq \min(m_2|E_2|, \frac{1}{2}m_2m_1 |V_2|\Delta_1),$$ 
which concludes the proof.
\end{proof}

Given an independent set in $\mathcal{C}$, \mar{Proposition~\ref{prop:reduction_alpha} states that} all similarity edges involved in the related $c_4$s constitute a matching. Consequently, 

the optimal value $\alpha(\mathcal{C})$ can be bounded using Lemma~\ref{vertexlemma} with $m_1=1$ and $m_2=1$. This leads immediately to the following bound:

\begin{corollary}\label{cor:alpha}
Given an instance $\prec G_1,G_2,S \succ$ with conflict graph $\mathcal{C}$, the independence number of $\mathcal{C}$  satisfies:\\ 
{\centerline{$\alpha(\mathcal{C})\leq \min \left(|E_1|,|E_2|, \frac{1}{2}|V_1|\Delta_2,\frac{1}{2}|V_2|\Delta_1\right)$.}}
\end{corollary}

The following lemma generalises the 
bound for degrees provided in~\citet{Fertin200990} for the case where $m_2=1$.

\begin{lemma}
\label{deglemma}
Given an instance $\prec G_1,G_2,S \succ$ with conflict graph $\mathcal{C}$, let $\gamma(x)=abcd$ be a $c_4$ corresponding to a vertex $x$ in  $\mathcal{C}$, then the degrees in $\mathcal{C}$ satisfy:

\noindent
$ \begin{array}{ll}
 {\rm\textbf{(i)}} &
d_{\mathcal{C}}(x) \leq  m_1(m_1-1)(d_1(a)+d_1(b))-(m_1-1)^2 + m_2(m_2-1)(d_2(c)+d_2(d))-(m_2-1)^2;\\
 {\rm\textbf{(ii)}} & \Delta(\mathcal{C})\leq  2\Delta_1{m_1}^2+2\Delta_2{m_2}^2-2\Delta_1m_1-2\Delta_2m_2-{m_1}^2-{m_2}^2+2m_1+2m_2-2.
 \end{array} $
\end{lemma}
\begin{proof} 
\textbf{(i)}
Denote the set of 
$c_4$s in $\mathcal{C}$ 
conflicting with $\gamma(x)$ 
with $S_1\cup S_2$, where $S_1$ is the set of 
$c_4$s 
in conflict with $\gamma(x)$ 
that include $ad$ or $bc$, and $S_2$ consists of all other $c_4$s conflicting with $\gamma(x)$.  
It is clear that, if a $c_4$ from $S_1$ shares the edge $ad\ (bc)$ with $\gamma(x)$, it must also include either $b\ (a)$ or $c\ (d)$ in order to create a conflict with $\gamma(x)$. In any case, since the total number of valid similarity edges (edges that can create the conflict with $\gamma(x)$) incident to $b$ and $c$ ($a$ and $d$) is bounded by $m_1+m_2-2$, this implies that $|S_1|$ is upper-bounded by $2m_1+2m_2-4$. 
For the second set $S_2$, we first note that a pair of similarity edges can create only one $c_4$. 
This implies that any edge in $G_1$
different from $ab$
can be part of at most ${m_1}^2-m_1$ different $c_4$s in $S_2$ and any edge in $G_2$ different from $cd$ can be part of at most ${m_2}^2-m_2$
different $c_4$s in $S_2$.
Since the number of $G_1$ edges incident to $a$ or $b$, and different from $ab$ is 
 $d_1(a)+d_1(b)-2$, and respectively the number of $G_2$ edges incident to $c$ or $d$, and different from $cd$ is 
at most $d_2(c)+d_2(d)-2$, the number of $c_4$s in $S_2$ that do not include $ab$ or $cd$
is bounded by $(d_1(a)+d_1(b)-2)({m_1}^2-m_1)+(d_2(c)+d_2(d)-2)({m_2}^2-m_2)$.
The edges $ab$ and $cd$ themselves can be part of at most $(m_1-1)^2$ and $(m_2-1)^2$ different $c_4$s in $S_2$ respectively, which concludes the proof of \textbf{(i)}. \textbf{(ii)} is immediately deduced since  $d_1(a), d_1(b) \leq \Delta_1$ and $d_2(c), d_2(d) \leq \Delta_2$.
\end{proof}

 \mar{When evaluating the number of edges of the conflict graph, the first Zagreb index of the graphs $G_1$ and $G_2$ appear naturally, as stated in the following lemma. Note that if $m_1=1$ (resp., $m_2=1$), then the bound only depends on $G_2$ (resp., $G_1$).} 

\begin{lemma}
\label{edgelemma}
Given an instance $\prec G_1,G_2,S \succ$ with conflict graph $\mathcal{C}$, the number $|E_\mathcal{C}|$ of edges of $\mathcal{C}$ is bounded by: 
$$|E_\mathcal{C}|\leq  \frac{1}{2}\left ({m_1}^2(m_1-1)\left (m_1M_1(G_1)-(m_1-1)|E_1|\right )+ {m_2}^2(m_2-1)\left (m_2M_1(G_2)-(m_2-1)|E_2|\right )\right ).$$
\end{lemma}
\begin{proof}
We have $|E_\mathcal{C}|=\frac{1}{2}\sum_{x\in V_{\mathcal{C}}}d_{\mathcal{C}}(x)$.
Using Lemma~\ref{deglemma} and the fact that $ab$ (resp., $cd$) participates to at most ${m_1}^2$ (resp., ${m_2}^2$) $c_4$s we get:
\begin{equation}\label{eq:edges}
\begin{array}{rccl}
2|E_\mathcal{C}|&\leq & &
{m_1}^2(m_1-1)\sum\limits_{ab\in E_1}[m_1(d_1(a)+d_1(b))-(m_1-1)]\\
&&+& {m_2}^2(m_2-1)\sum\limits_{cd\in E_2}[m_2(d_2(c)+d_2(b))-(m_2-1)]\\
&\leq && m_1^3(m_1-1)\sum\limits_{ab\in E_1}(d_1(a)+d_1(b))-{m_1}^2(m_1-1)^2|E_1|\\ 
&&+& m_2^3(m_2-1)\sum\limits_{cd\in E_2}(d_2(c)+d_2(d))-{m_2}^2(m_2-1)^2|E_2|
\end{array}
\end{equation}
We conclude by noting that  $\sum_{ab\in E_1}(d_1(a)+d_1(b))=M_1(G_1)$  and similarly for $cd$ in the graph $G_2$. \end{proof}

\mar{If we want a bound for $|E_\mathcal{C}|$ only dependent on the degree, number of vertices and edges of $G_1, G_2$, then  several upper bounds exist for the first Zagreb index.  We mention here two of these bounds.}

\begin{theorem}\label{th:zagreb-bound} Given a connected graph $G=(V,E)$ with maximum degree $\Delta$ and minimum degree~$\delta$,\\
\textbf{ (i)}~(\citet{Zagreb-Liu}) $M_1(G)\leq \frac{|E|^2(\Delta+\delta)^2}{n\Delta \delta}$;\\
\textbf{ (ii)}~(\citet{Zagreb-Tabar})
$M_1(G)\leq 4\frac{|E|^2}{|V|}+ \frac{|V|}{4}(\Delta-\delta)^2$.
\end{theorem}

Note that the bound $M_1(G)\leq 2\Delta|E|$ is trivial for all graph $G=(V,E)$ and with maximum degree $\Delta$. This bound meets the two bounds in Theorem~\ref{th:zagreb-bound} for regular graphs. In Subsection~\ref{subsec:acyclic}, we will consider the class of acyclic graphs. For this class ($\delta=1$ and $|E|\leq |V|$), the bound \textbf{ (i)} immediately gives $M_1(G)\leq |E|\frac{(\Delta+1)^2}{\Delta} \leq |E|(\Delta +3)$, thus twice better than the trivial bound. Note also that in the case where one of these graphs has much less edges,  $|E_1|\mar{\in} o(|E_2|)$ or $|E_2|\mar{\in} o(|E_1|)$, then a direct application of Lemma~\ref{vertexlemma}, using  $|E_\mathcal{C}|\leq |V_\mathcal{C}|^2$, can give better bounds. 

\mar{Lemma~\ref{edgelemma}  and Theorem~\ref{th:zagreb-bound} will be used in Subsection~\ref{subsec:acyclic}. Below we provide direct consequences of Lemmas~\ref{vertexlemma} and~\ref{deglemma} leading to the design of polynomial-time approximation algorithms for the constrained alignment problem.}


The best known approximation ratios guaranteed by polynomial algorithms   for the maximum independent set problem are  $O(\Delta\log \log\Delta/\log \Delta)$~(\citet{approx-stable}) and $O(n/\log^2 n)$~(\citet{Boppana92}), where $\Delta$ and $n$ denote respectively the maximum degree and the number of vertices of the input graph. 
Combining it with
Lemmas~\ref{vertexlemma} and~\ref{deglemma}   leads to the following approximation for the general setting.

\begin{proposition}\label{prop:generalapprox}\mbox{}\\
\textbf{ (i)} For any positive constant $m_1$, $m_2$, 
the constrained alignment problem can be approximated in polynomial time with an approximation ratio of
$O((\Delta_1+\Delta_2)\log \log(\Delta_1+\Delta_2)/\log(\Delta_1+\Delta_2))$;\\
\textbf{ (ii)} If only $m_2$ (resp. $m_1$) is constant, then the constrained alignment problem can be approximated in polynomial time with an approximation ratio of
$O(|E_1|/\log^2|E_1|)$ (resp. $O(|E_2|/\log^2|E_2|)$).
\end{proposition}

It is known that using bounded search techniques~(\citet{ParameterizedComplexity}),
one can find an independent set of size $k$ in a graph $G$ in $O(n(\Delta(G)+1)^k)$ time, or return that no such subset exists.
In~\citet{Fertin200990}, 
this result  is used to show that the constrained alignment problem 
is fixed-parameter tractable for bounded degree graphs with $m_2=1$. Lemma~\ref{deglemma} immediately provides a generalisation for the general setting.

\begin{proposition}\label{prop:parameter1}
Provided that $G_1$ and $G_2$ are bounded degree graphs, for any positive constants $m_1$, $m_2$, 
the constrained alignment problem is fixed-parameter tractable for parameter $k$ and 
solvable in\\ $O(min(|E_1|,|E_2|)(D+1)^k)$ time, where $k$ is the number of final conserved edges and $D=O(\Delta_1+\Delta_2)$.
\end{proposition}

In what follows we consider the case $m_2=1$ \-- which, \mar{to our knowledge}, is the most studied case \-- and investigate specific properties of the conflict graph. \mar{This case, by itself already very hard, simplifies the possible conflicts and then perfectly illustrates the use of the conflict graph. As explained in the conclusion, the following results motivate the further study of conflict graphs and their independent sets for a more general set-up. }

\section{The case $m_2=1$}\label{sec:m2is1}

The case with $m_2=1$ is the main case considered in~\citet{Fertin200990}.  We remind that, in this case,  the possible conflicting configurations are listed in Figure~\ref{twoc4}.  Some improved results deal with the particular case $m_1=2$. 
It is 
known that the problem is APX-hard even for the case where
$m_1=2$ and both $G_1, G_2$ are bipartite~(\citet{Fertin200990}).

\subsection{Structure of $\mathcal{C}$ and approximation}\label{subsec:struct-approx}

In this subsection we present graph theoretic properties of conflict graphs in terms of forbidden subgraphs when $m_1=2$. 
In addition to providing valuable information regarding 
structural properties of conflict graphs, it has also algorithmic applications, mainly approximation results. 

Note first that, if $m_2=1$, Lemma~\ref{deglemma} states that the maximum degree of the conflict graph is at most $2({m_1}^2-m_1)\Delta_1 +m_1(2-m_1)-1$ and consequently Proposition~\ref{prop:generalapprox} can be immediately replaced by:

\begin{proposition}\label{cor:approxD1}
For $m_2=1$ and any positive constant $m_1$, the constrained alignment problem can be approximated in polynomial time with an approximation ratio of
$O(\Delta_1\log \log(\Delta_1)/\log(\Delta_1))$.
\end{proposition}
 
This approximation ratio in $o(\Delta_1)$ improves the result of
~\citet{Fertin200990} \-- $2\lceil3\Delta_1/5\rceil$
for even $\Delta_1$ and $2\lceil(3\Delta_1+2)/5\rceil$ for odd 
$\Delta_1$ \-- also obtained for $m_2=1$. We will give later another improvement in the case where $\Delta_2$ is less than this ratio. 

\mar{We first establish some  properties of conflict graphs when $m_2=1$ \-- Facts~\ref{onenode} and~\ref{twonodes}, Lemmas~\ref{lem_H} and~\ref{lem:5nodes} and Corollary~\ref{lem_P4} \--  that will be useful for the main structural and algorithmic results. Then, in paragraphs~\ref{subsub:wf} and~\ref{subsub:cc}, we derive structural results and their algorithmic consequences. }

\begin{fact}
\label{onenode}
Any pair of conflicting $c_4$s in $\mathcal{C}_U$ must share at least one vertex from $G_1$.
\end{fact}

\begin{fact}
\label{twonodes}
Any pair of distinct $c_4$s in $\mathcal{C}_U$ sharing two vertices from $G_1$ has a conflict.
\end{fact}

\begin{lemma}
\label{lem_H}
Given an instance $\prec G_1,G_2,S \succ$ with conflict graph $\mathcal{C}$, suppose $m_2=1$ and consider an induced subgraph $H$  of $\mathcal{C}$ such that $\overline{H}$ is connected and $H$ has an induced $P_3$.
Then the $c_4$s in $H$ cannot all share a vertex from $G_1$.
\end{lemma}

\begin{proof}
Let \mar{$x_1x_2x_3$} be an induced $P_3$ in $H$  and let $\gamma(x_1)=abcd$. 
Assume for the sake of contradiction that $a\in G_1$ is a vertex common to all the $c_4$s \mar{associated with vertices of} $H$. For every two vertices $y,z$ in $H$ not linked by an edge, $\gamma(y)$ and $\gamma(z)$ must share the similarity edge including $a$ to avoid any conflict. As a consequence and since  $\overline{H}$ is connected, all 
the $c_4$s \mar{associated with vertices of} $H$ must share the edge $ad$. 
This implies that  any conflict between any pair of these $c_4$s can only be either a
$Conf_{3a}$ or a $Conf_{3b}$ conflict, which further implies that all the $c_4$s $\gamma(x_i), 
\mar{i=1,2,3}$ include  $b$.
By Fact~\ref{twonodes} and since $x_1\neq x_3$, this implies a conflict between $\gamma(x_1)$ and $\gamma(x_3)$, a contradiction.
\end{proof}

For instance, a $P_4$ or $P_3+K_1$ \-- the independent union of a $P_3$ and an isolated vertex \-- clearly both satisfy the conditions on $H$: they both have an induced $P_3$ and moreover, $\overline{P_4}$ is a $P_4$ as well while $\overline{P_3+K_1}$ is a triangle with a pendent vertex, both connected. So, we immediately deduce:

\begin{corollary}
\label{lem_P4}
Given an instance $\prec G_1,G_2,S \succ$ with conflict graph $\mathcal{C}$, if $m_2=1$, the four $c_4$s of an induced $P_4$ or an induced $P_3+K_1$ of $\mathcal{C}$ 
cannot all share a vertex from $G_1$.
\end{corollary}

The following lemma will be useful for studying the structure of $\mathcal{C}$.

\begin{lemma}\label{lem:5nodes}
Given an instance $\prec G_1,G_2,S \succ$ with conflict graph $\mathcal{C}$, suppose $m_2=1$ and that  we have in $\mathcal{C}$ an induced $P_5$ $x_1x_2x_3x_4x_5$ as well as two vertices $y_1,y_2$ not linked to $x_i, i=2,3,4$ and an additional vertex  $x$ linked to the seven vertices $y_1,y_2$, $x_1,x_2,x_3,x_4, x_5$. Denote $\gamma(x)=abcd$.
Then if $\gamma(y_1)$ does not include $b$, neither does $\gamma(y_2)$. 
\end{lemma}

\begin{proof}
Since $y_1,y_2,x_1,x_2,x_3, x_4, x_5$ are all linked to $x$, Fact~\ref{onenode} ensures the related $c_4$s include  $a$ or $b$. 
Assume for the sake of contradiction that $\gamma(y_1)$ does not include $b$ while $\gamma(y_2)$  
\mar{does}. Since $\gamma(y_1)$ conflicts \mar{with} $\gamma(x)$, we have  $\gamma(y_1)=aklm$ with $k \in V_1 \setminus \{a,b\}$ and $m\neq d$. Let  $\gamma(y_2)=bpqr$, $r\neq c$. 
Since $m_2=1$, 
 $m, d,c,r$ are all paire wise distinct.

As mentioned above $\gamma(x_j), j=1,\ldots 5$ must include $a$ or $b$. Since $\gamma(x_j), j=2,3,4$ do not conflict with $\gamma(y_1)$ nor with $\gamma(y_2)$, if it includes $a$ it must include the edge $am$
and if it includes $b$ it must include the edge $br$. Moreover, none of them can include both $a$ and $b$. Indeed, in this case  
$\gamma(x_j)=abrm$ for some $j=2,3,4$ and since any $\gamma(x_{j'})$,  $j'\in \{2,3,4\}\setminus\{j\}$, can neither include an edge $am'$, $m'\neq m$ nor $br'$, $r'\neq r$, it cannot conflict with $\gamma(x_j)$, a contradiction.

On the other hand, since $\gamma(x_3)$ has a conflict with both $\gamma(x_2)$ and $\gamma(x_4)$ and since  $\gamma(x_2)$ and $\gamma(x_4)$ are not conflicting, there must be two similarity edges $uv$, $uv'$, $u\in V_1\setminus\{a,b\}$, $v,v'\in V_2, v\neq v'$, where $uv$ is an edge of $\gamma(x_3)$ and $uv'$ is an edge of both $\gamma(x_2)$ and $\gamma(x_4)$. Since $\gamma(x_2)\neq \gamma(x_4)$, one of them includes the edge $am$ and the other includes the edge $br$.

We consider below the possible cases that all lead to a contradiction.

 \emph{Case-1:}
Suppose $\gamma(x_2)=auv'm$ and $\gamma(x_4)=buv'r$, thus   
$\gamma(x_3)$ is either $auvm$ or $buvr$. 

  \emph{Case-1.1:} If $\gamma(x_3)=auvm$, then since $\gamma(x_1)$ conflicts with $\gamma(x_2)$ but not with $\gamma(x_3)$ it must include the edge $uv$ but in this case it would conflict with $\gamma(x_4)$.
 
  \emph{Case-1.2:} Similarly if $\gamma(x_3)=buvr$, then since $\gamma(x_5)$ conflicts with $\gamma(x_4)$ but not with $\gamma(x_3)$ it must include the edge $uv$ but in this case it would conflict with $\gamma(x_2)$.

  \emph{Case-2:}
Suppose now  $\gamma(x_2)=buv'r$ and $\gamma(x_4)=auv'm$, thus   
$\gamma(x_3)$ is either $auvm$ or $buvr$.
In both cases we get the same contradiction as in Case-1 exchanging the roles of $am$ and $br$. This concludes the proof.
\end{proof}

\subsubsection{Wheels and Fans}\label{subsub:wf}

\begin{theorem}
\label{lemW8} 
Given an instance $\prec G_1,G_2,S \succ$ with conflict graph $\mathcal{C}$,\\
\textbf{ (i)}
If $m_2=1$, 
$\mathcal{C}$ is $W_t$-free, for $t\geq 7$;\\
\textbf{ (ii)} If furthermore $m_1=2$, $\mathcal{C}$ is also $W_5$ and $W_6$-free. 
\end{theorem}

\begin{proof}
Assume for the sake of contradiction an induced $W_t$ exists with  $t\geq 5$ and 
let $x$ be the center vertex with $\gamma(x)=abcd$. Let $x_1x_2\ldots x_tx_1$
be the induced $C_t$ of the wheel $W_t$ in the conflict graph.
By Fact~\ref{onenode}
every $\mar{\gamma(x_i)}, 1\leq i\leq t$ must include at least one of 
$a$ or $b$. 
By Corollary~\ref{lem_P4} (the cycle $C_t$ has an induced $P_4$), it is not possible for all of 
\mar{these $c_4$s} to share 
$a$, nor can they all share $b$. This implies that there must exist 
a pair of conflicting $c_4$s, $\gamma(x_i), i=1, \ldots, t$, such that their corresponding vertices in $\mathcal{C}$ are neighbors in $C_t$, one including $a$ and the other including $b$ and one of them does not contain both $a$ and $b$. 
Without loss of generality, let the former be $\gamma(x_t)=aklm$  with $k \in V_1 \setminus \{a,b\}$ and the latter be $\gamma(x_{t-1})=bpqr$. 


\textbf{(i)}
Assume first $t\geq 7$. Then apply Lemma~\ref{lem:5nodes} with $y_1=x_t$ and $y_2=x_{t-1}$ gives a contradiction.

\textbf{(ii)} Now we show directly that there is also a contradiction if $m_1=2$ and $t=5,6$. 

We consider two cases $\gamma(x_{t-1})=abrd$,  and $\gamma(x_{t-1})=bkl'r$, $l'\neq l$ ensuring the conflict between $\gamma(x_{t-1})$ and $\gamma(x_{t})$. In both cases $r\neq c$ ensures the conflict between $\gamma(x_{t-1})$ and $\gamma(x)$.

 \emph{Case-1:} $\gamma(x_{t-1})=abrd$. Since $\gamma(x_2), \gamma(x_1)$ have no conflict with $\gamma(x_{t-1})$ but with $\gamma(x)$, they both include $br$ and not $am$. Moreover, since $\gamma(x_2)$ does not conflict  $\gamma(x_{t})$, it cannot include $a$ and thus $\gamma(x_2)=buvr, u\neq a$.  Since  $\gamma(x_1)$  conflicts with both $\gamma(x_t)$ and $\gamma(x_2)$ we have $u=k$, $v=l$ and $\gamma(x_1)=bkl'r, l'\neq l, \gamma(x_2)=bklr$. Then, since $\gamma(x_3)$ conflicts with $\gamma(x_2)$
but not with $\gamma(x_1)$, it must include the edge $kl'$. To conflict with $\gamma(x)$ it should include $am$ or $br$, a contradiction since $x_3\neq x_t, x_3\neq x_1$.

 \emph{Case-2:} $\gamma(x_{t-1})=bkl'r$, $l'\neq l$.

 Since $\gamma(x_{1})$ conflicts with $\gamma(x_{t})$ and with $\gamma(x)$ and since $x_1\neq x_{t-1}$, $\gamma(x_{1})$ cannot include $br$ and thus includes $am$ and  $\gamma(x_{1})=akl'm$.

$\gamma(x_{t-2})$ conflicts with $\gamma(x_{t-1})$ but not with $\gamma(x_{t})$ and includes  $am$ or $br$.
Since $x_{t-2}\neq x_t$ the only possibility is 
$\gamma(x_{t-2})=bklr$. 
Then, $\gamma(x_{2})$ conflicts with $\gamma(x_1)$ but not with $\gamma(x_{t})$ and includes  $am$ or $br$; the two only candidates are $aklm$ and $bklm$, both impossible since $x_2\neq x_t, x_2\neq x_{t-2}$ (note that $t-2\geq 3$). It concludes the proof.
\end{proof}

\begin{figure}[t]	   
\begin{center}	   
\includegraphics[width=12cm]{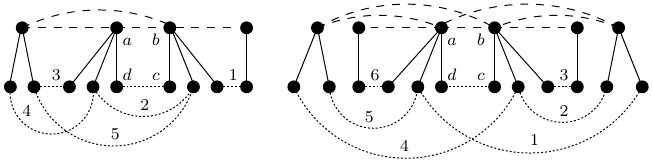} 
\caption{\sf Sample configurations for $\mathcal{C}_U$s inducing $W_5$ (left) and $W_6$ (right) in their respective 
conflict graphs
for the case where $m_1=3$.
The central vertices of the wheels in each case correspond to the $c_4$s  
indicated with $abcd$.
} 
\label{w7sample}	   
\end{center}	   
\end{figure}  

Note that for $m_1>2$, it is still possible to have a $W_5$ and 
$W_6$ in a conflict graph  as illustrated in Figure~\ref{w7sample}. Note also that $W_4$  and $w_3=K_4$ can still exist in $\mathcal{C}$ even if $m_1=2$. Figure~\ref{cor8fig} gives a sample construction with a $W_4$ while Figure~\ref{cliquesample} gives an example with a $K_4$. 
It means that, in terms of induced wheels, Theorem~\ref{lemW8}  leaves no gap.

  \begin{figure}[t]     
\begin{center}	   
\includegraphics[width=12cm]{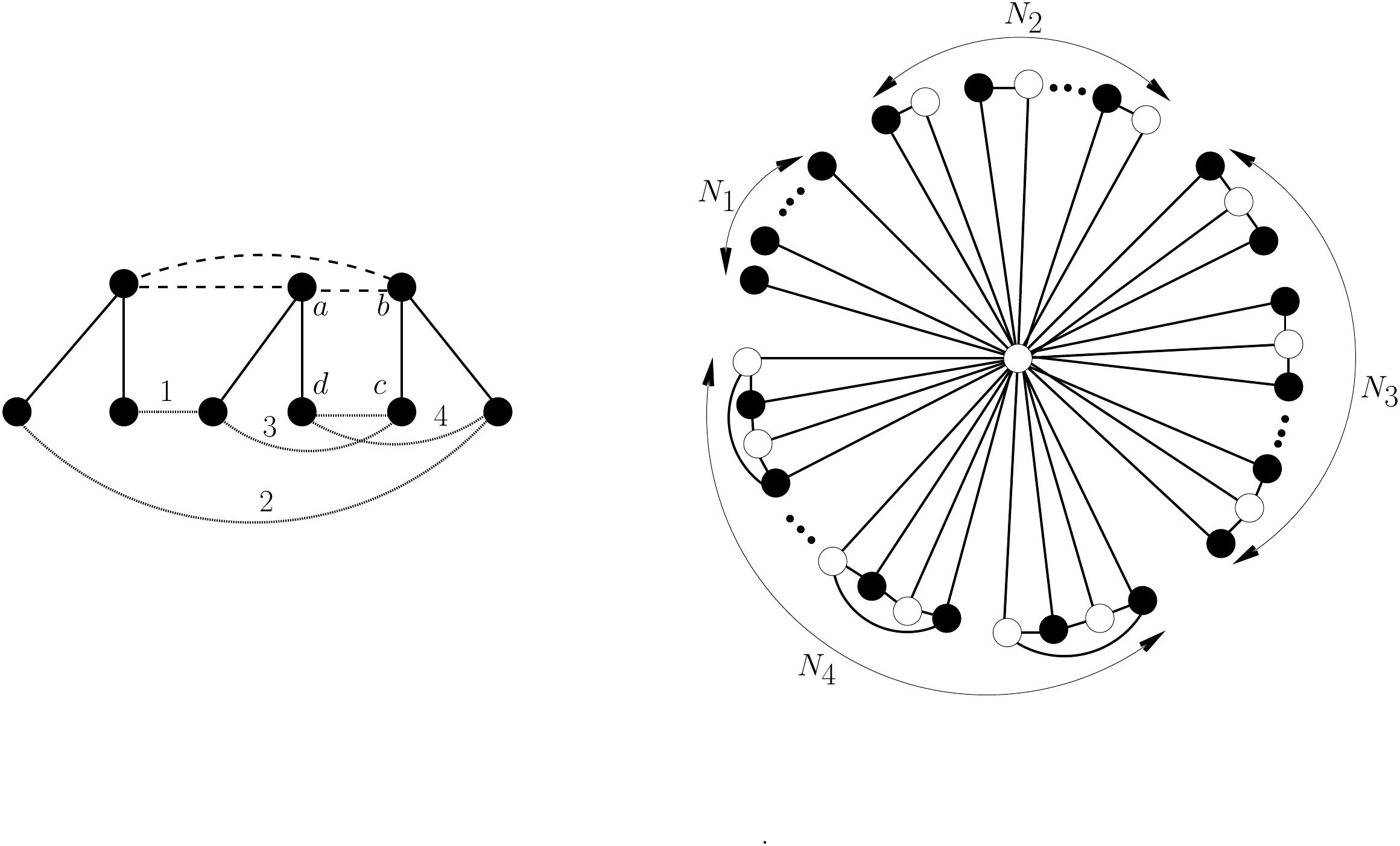} 
\caption{\sf \textbf{Left:} Sample construction for a $W_4$ in $\mathcal{C}$. The central vertex of the wheel corresponds to the $c_4$  $abcd$. The upper partition corresponds to vertices of $G_1$ and the lower partition to those of $G_2$. Similarity edges are drawn between the partitions. \textbf{ Right:} Depiction of the construction defined in Lemma~\ref{lem:F8}. The vertex $x$  is shown in the center and vertices in $S_x^1$ are shown at the peripheral. $N_t$ corresponds to all vertices in components of $\mathcal{C}[S_x^1]$ of size~$t$. The black vertices constitute a maximum independent set of $\mathcal{C}[S_x^1]$. 
} 
\label{cor8fig}	   
\end{center}	   
\end{figure}

\begin{figure}[t]	   
\begin{center}	   
\includegraphics[width=14cm]{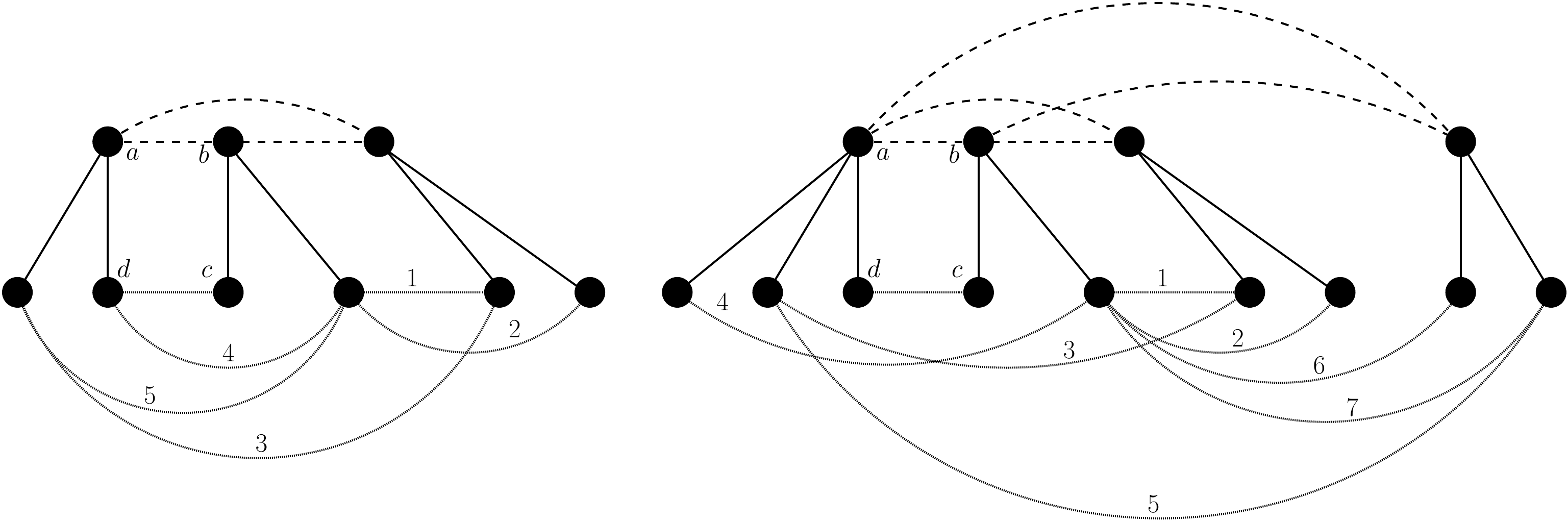} 
\caption{\sf \textbf{ Left:} Sample configuration for $\mathcal{C}_U$ inducing
$F_5$ for the case where $m_1=2$. 
\textbf{Right:} Sample configuration for $\mathcal{C}_U$ inducing $F_7$ 
for the case where $m_1=3$. In each case 
the central vertex corresponds to the $c_4$  
indicated with $abcd$. 
Each $G_2$ edge is marked with the related $c_4$  in $P_5=12345$ (left) or $P_7=1234567$ (right). 
} 
\label{f5and7sample}	   
\end{center}	   
\end{figure}

The following lemma gives an example how considering the different kind of conflicts, for $m_2=1$ and $m_1=2$, (see Figure~\ref{twoc4}) helps understanding the structure of the conflict graph. 

\begin{lemma}\label{lem:F8}
Given an instance $\prec G_1,G_2,S \succ$ with conflict graph $\mathcal{C}$, suppose $m_2=1$ and $m_1=2$ and consider a vertex $x$ in $\mathcal{C}$ and the set $S_x^1$ of $c_4$s that conflict $\gamma(x)$ with a $Conf_{1a}$ or $Conf_{1b}$ configuration. Then, $\mathcal{C}[S_x^1]$ is an independent collection of $C_4$s, $P_3$s, $P_2$s and isolated vertices.\end{lemma}

\begin{proof}
Since $m_1=2$, in $c_4\in S_x^1$, at most two $c_4$s   can conflict a fixed $c_4$ and consequently the graph $\mathcal{C}[S_x^1]$ has degree at most~2, which means it is an independent  collection of cycles and paths. For any $t\geq 1$, consider a connected component of $\mathcal{C}[S_x^1]$ of size~$t$. 

Assume we have $u_1,u_2,u_3$ in $\mathcal{C}[S_x^1]$ with edges $u_1u_2$ and $u_2u_3$. Since $m_1=2$ $\gamma(u_1),\gamma(u_2)$ and $\gamma(u_3)$  cannot all include $a$ and neither can they all include $b$.  Suppose without loss of generality that two of them include $a$ and one $b$ and in this case, the structure of conflicts $Conf_{1a}$ and $Conf_{1b}$ imposes that $\gamma(u_2)$ includes $a$, say $\gamma(u_2)=aklm$ with $k,l,m\notin \{b,c,d\}$. Suppose then without loss of generality that $\gamma(u_1)$ includes $a$ and $\gamma(u_3)$ includes $b$: $\gamma(u_1)=akl'm$, $l'\neq l$ and necessarily $\gamma(u_3)=brl'k$ to create a conflict with $\gamma(u_2)$. Moreover, since $m_2=1$, $r\notin \{c,d,m,l,l'\}$.

Note then that we cannot have any conflict between $\gamma(u_3)$ and $\gamma(u_1)$, which means that $\mathcal{C}[S_x^1]$ is triangle-free. Moreover suppose a fourth $c_4$, $\gamma(u_4)$ conflicting $\gamma(u_3)$ in $\mathcal{C}[S_x^1]$. It necessarily includes $kl$ and thus conflicts $\gamma(u_1)$, which means that $\mathcal{C}[S_x^1]$ is $P_4$-free, which completes the proof. Figure~\ref{cor8fig} (Right) describes the structure of $\mathcal{C}[S_x^1]$, where $N_t$, $t=1,2,3,4$, is the union of components of $\mathcal{C}[S_x^1]$ of size~$t$.
\end{proof}

\begin{corollary}\label{polalg1}
Given an instance $\prec G_1,G_2,S \succ$ with conflict graph $\mathcal{C}$, if $m_1=2$ and $m_2=1$, then for every $x\in V_{\mathcal{C}}$, removing at most two vertices to $\mathcal{C}_x$ makes it an independent collection of $C_4$s, $P_3$s, $P_2$s and isolated vertices. 
\end{corollary}
\begin{proof} 
If $m_1=2$, at most one $c_4$ conflicts $\gamma(x)$  with a  $Conf_{3a}$  configuration, and 
at most one with  $Conf_{3b}$ configuration. Let us remove these vertices.  There can be at most one $c_4$ conflicting $\gamma(x)$ with a  $Conf_{2}$ configuration and moreover such a $c_4$ necessarily \mar{corresponds to} an isolated vertex in $\mathcal{C}[S_x^1]$. Since all the other neighbors of $x$ correspond to   $Conf_{1a}$ or $Conf_{1b}$ configurations,  Lemma~\ref{lem:F8} immediately concludes the proof.   
\end{proof}

\begin{corollary}\label{rem:degree}
If $m_1=2$ and $m_2=1$ we are ensured to find in polynomial time a legal alignment with at least 
$(\Delta(\mathcal{C})-2)/2$ conserved edges.
\end{corollary}

\mar{\begin{proof}
It is an immediate consequence of Corollary~\ref{polalg1} applied to a vertex $x$ of maximum degree in $\mathcal{C}$. An exhaustive search or just the detection of $Conf_{3a}$  and 
$Conf_{3b}$ configurations involving $\gamma(x)$ allows to identify the vertices to be removed to make $\mathcal{C}_x$ an independent collection of $C_4$s, $P_3$s, $P_2$s and isolated vertices. Picking in this collection an independent set of two vertices in each   $C_4$s and $P_3$, one vertex in each $P_2$ and all the isolated vertices gives a independent set of size at least~$(\Delta(\mathcal{C})-2)/2$ in $\mathcal{C}$. Using Proposition~\ref{prop:reduction_alpha} and the fact that the function $\gamma$ can be computed in polynomial time (see Corollary~\ref{cor:reduction}) allows to conclude. 
\end{proof}}

The following result concerns the existence of induced fans $F_t$ in the conflict graph. Note that for $2\leq t\leq t'$, $F_t$ is an induced subgraph of $F_{t'}$ and consequently an $F_t$-free graph is also $F_{t'}$-free. 


\begin{theorem}
\label{Th:F8} 
 Given an instance $\prec G_1,G_2,S \succ$ with conflict graph $\mathcal{C}$ such that $m_2=1$, then:\\
\textbf{ (i)}
For $m_1\geq 3$, $\mathcal{C}$ is $F_8$-free;\\
\textbf{ (ii)} For $m_1=2$, $\mathcal{C}$ is $F_6$-free.
\end{theorem}

\begin{proof}
Consider an induced $F_t$ and let $\gamma(x)=abcd$ be the center vertex. 
 
 \textbf{(i)} Assume for the sake of contradiction that $t=8$ and
 denote by \mar{$z_1z_2 \cdots z_8$}
be the induced $P_8$ in the neighborhood of $x$.
By Fact~\ref{onenode}
every $c_4$ $\gamma(z_i), i=1, \ldots, 8$ must include at least one of 
$a$ or $b$ 
in $\mathcal{C}_U$. 

Suppose first $\gamma(z_1)\cap \{a,b\}=\gamma(z_8)\cap \{a,b\}$. Without loss of generality we assume they both include $b$ and either both include $a$ as well or none of them. Consider then the subgraph induced by $z_1,z_2,z_3,z_8$, inducing a $P_3+K_1$.
By Corollary~\ref{lem_P4}, the $c_4$s $\gamma(z_1),\gamma(z_2),\gamma(z_3)$ and $\gamma(z_8)$  cannot all include $b$ and let $i\in\{2,3\}$ such that   $\gamma(z_i)$ does not include $b$. Then, Lemma~\ref{lem:5nodes} with $y_1=z_1$ and $y_2=z_i$ and $z_4, \ldots z_8$ corresponding to $x_1, \ldots x_5$ leads to a contradiction. 

Suppose now $\gamma(z_1)\cap \{a,b\}\neq \gamma(z_8)\cap \{a,b\}$, then one only includes $b$ and we get a contradiction as well by applying Lemma~\ref{lem:5nodes} with $y_1=z_1$ and $y_2=z_8$ and $z_2, \ldots z_6$ corresponding to $x_1, \ldots x_5$, which concludes the proof of \textbf{ (i)}.

\textbf{(ii)}
Assume now $m_1=2$. Corollary~\ref{polalg1} immediately shows that it is possible to remove at most two neighbors of $x$ so that $x$ cannot be the center of a $F_4$. It excludes the possibility of a $F_6$ in this case. 
 \end{proof}

Figure~\ref{f5and7sample}-Left shows an example of $F_5$ in a conflict graph with $m_1=2$ and $m_2=1$ and Figure~\ref{f5and7sample}-Right shows an $F_7$ in a conflict graph with $m_1=3$ and $m_2=1$.

Theorems~\ref{lemW8} and~\ref{Th:F8} as well as Corollary~\ref{polalg1}  give us information about the structure of the subgraphs $\mathcal{C}_x$, $x\in V_{\mathcal{C}}$, induced by  $N[x]$: as already mentioned  a graph $G$ is $W_t$-free (resp. $F_t$-free) if for all vertex $x$, $G_x$ is $C_t$-free (resp. $P_t$-free), two classes of graphs that raised a lot of interest from researchers (see, eg., \citet{isgci,graphclassesbook}). 

We give now an example how to use the structure of neighborhoods to approximate the maximum independent set problem. It will give us algorithmic  applications of  Corollaries~\ref{polalg1} and~\ref{lem_P4}. 

A very classical approximation algorithm for maximum independent set in a graph $G=(V,E)$ is the  algorithm {\tt 2-opt} determining an independent set $\tilde S$ such that $\forall u \in \tilde S, \forall v,w \in V\setminus \tilde S, (\tilde S\setminus \{u\})\cup \{v,w\}$ is not and independent set (there is no 2-improvement). Let us revisit the very usual analysis of {\tt 2-opt} (see, e.g.,~\citet{ejorapprox})  which consists in considering the bipartite graph $B$ induced by $\tilde S \cup S^*$, where $S^*$ is an optimum independent set. Denote by $\lambda(G)=|\tilde S|$ the value of the solution provided by the algorithm on $G$ and 
$\alpha(G)=|S^*|$ the independent number of $G$. Then the number of edges of $B$ is at least $2\alpha(G) - \lambda(G)$ since 2-optimality ensures that, for every two edges $\tilde vu, \tilde v w$ in $B$ incident to the same vertex   $\tilde v\in\tilde S$, there is an additional edge incident to $u$ or $v$. On the other hand this number is at most $\Delta_{\alpha}\lambda(G)$, where $\Delta_{\alpha}$ is
the minimum among all optimal independent sets $S$ of the maximum number of vertices in $S$ a vertex can be adjacent to:
$$\Delta_{\alpha}=\min_{\substack{|S|=\alpha(G)\\ S\ {\rm independent}}} \max_{v\in V}|N(v)\cap S|.$$  
This implies:
\begin{equation}\label{eq:2opt}
\frac{\alpha(G)}{\lambda(G)}\leq \frac{(\Delta_{\alpha}+1)}{2}.
\end{equation}
 This remark emphasises that the usual maximum degree can actually be replaced by $\Delta_{\alpha}$.
 We propose below a strategy that can be used where large independent sets can be found in polynomial time in the neighborhood of each vertex. It leads to a new kind of approximation ratios depending on the independence number.
 \begin{theorem}\label{th:approx}
Consider a class of graphs $\mathcal{G}$ for which there is a polynomial time algorithm $A$ approximating the maximum independent set problem within the ratio $\rho$ for every graph $G_x$, where $G=(V,E)\in  \mathcal{G}$ and  $x\in V$.\\
 Then  the maximum independent set problem can be approximated within $\sqrt{3\rho(G)\alpha(G)/4}$.\end{theorem}
 \begin{proof}
 The strategy, for an input graph $G=(V,E)$ in $\mathcal{G}$ is as follows: 
 \begin{quote}
  Apply $A$ in all subgraphs $G_x, x\in V$;\\ 
  Compute also a {\tt 2-opt}-solution;\\
 Take the best solution among the $|V|+1$ different solutions obtained.
  \end{quote}
  Note first that, if $\alpha(G)\leq 2$, then {\tt 2-opt} finds an optimal solution, so we assume $\alpha(G)\geq 3$.
  
Suppose first that $\Delta_\alpha> \sqrt{4\rho(G)\alpha(G)/3}$. Then, when applied to a graph $G_x$ such that $\alpha(G_x)=\Delta_\alpha$, the algorithm $A$ computes a solution of value at least $\sqrt{4\alpha(G)/(3\rho(G))}$ leading to the approximation ratio  $\sqrt{3\rho(G)\alpha(G)/4}$.

Suppose now $\Delta_\alpha\leq \sqrt{4\rho(G)\alpha(G)/3}$, then Relation~(\ref{eq:2opt}) gives the ratio: 
$$\frac{\sqrt{\frac{4}{3}\rho(G)\alpha(G)}+1}{2}
\leq 
\sqrt{\rho(G)\alpha(G)}\frac{\frac{2}{\sqrt{3}}+\frac{1}{\sqrt{3}}}{2}=
\sqrt{\rho(G)\alpha(G)}\frac{\sqrt{3}}{2}
$$
where the inequality uses $\rho(G)\alpha(G)\geq\alpha(G)\geq 3$. 
In all cases, the ratio is at most $\sqrt{3\rho(G)\alpha(G)/4}$,
 which concludes the proof.
\end{proof}

Given an instance $I=\prec G_1,G_2,S \succ$, we denote by $\beta(I)$ the optimal value of the constrained alignment problem on $I$. 
\begin{proposition}\label{pro:approxsqrt}
Given an instance $\prec G_1,G_2,S \succ$ with conflict graph $\mathcal{C}$ and $m_2=1$,\\
\textbf{(i)} The constrained alignment problem can be approximated within $\sqrt{3\beta(I)/2}$;\\
\textbf{(ii)} If furthermore $m_1=2$, this is improved to  $\sqrt{\beta(I)}$.
\end{proposition}
\begin{proof}

This is a direct application of Theorem~\ref{th:approx}.\\
\textbf{(i)} Consider a vertex $x$ in the conflict graph $\mathcal{C}$ and the graph $\mathcal{C}_x$. We denote $\gamma(x)=abcd$. Using Fact~\ref{onenode},   
the $c_4$s in the neighborhood of $x$ in $\mathcal{C}$ can be partitioned into $\mar{N}_{x,a}$ and $\mar{N}_{x,b}$, where all $c_4$s in $N_{x,a}$ include $a$ while the others include $b$ but not $a$. This partition can be determined in polynomial time.  
Corollary~\ref{lem_P4} ensures that $\mathcal{C}[N_{x,a}]$ and $\mathcal{C}[N_{x,b}]$ are $P_4$-free. It is well known that the maximum independent set problem can be solved in linear time in $P_4$-free graphs (also called {\em cographs}) (see, e.g., \citet{golumbicbook}). Determining a maximum independent set in $\mathcal{C}[N_{x,a}]$ and $\mathcal{C}[N_{x,b}]$ and choosing the best one clearly solves the maximum independent set problem in  
\mar{$\mathcal{C}_x$} within an approximation ratio of~2. We apply Theorem~\ref{th:approx} with constant $\rho(G)=2$.\\
\textbf{(ii)} If $m_1=2$, then Corollary~\ref{polalg1} ensures that a maximum independent set can be found in polynomial time in graph $\mathcal{C}[N_{x}]$ and we apply Theorem~\ref{th:approx} with constant $\rho(G)=1$.
\end{proof}

Note that we obtain a ratio depending on the optimal value, which is not usual. Roughly speaking this result means that the logarithmic version of the problem \-- where the objective is to maximise the logarithm of the number of similarities in a legal alignment \-- is $\frac{3}{2}$-approximable. For instance, such a ratio for the maximum independent set in conflict graphs cannot be achieved in general graphs: the usual $n^{1-\varepsilon}$-hardness result~ (\citet{hastad}) states that, under some complexity hypothesis, the logarithm of the independence number cannot be approximated within a constant ratio.

Combining Proposition~\ref{pro:approxsqrt} with Corollary~\ref{cor:alpha} leads to the following ratio:

\begin{proposition}\label{cor:approxalpha2}
Given an instance $\prec G_1,G_2,S \succ$ with conflict graph $\mathcal{C}$ and $m_2=1$,\\
{\bf{(i)}} The constrained alignment problem can be approximated within the ratio:\\
{\centerline{$\min \left(\sqrt{3/2}\sqrt{|E_1|},\sqrt{3/2}\sqrt{|E_2|}, (1/2)\sqrt{3|V_1|\Delta_2},(1/2)\sqrt{3|V_2|\Delta_1}\right)$;}}
{\bf{(ii)}} If furthermore $m_1=2$, this ratio becomes:\\
{\centerline{$\min \left(\sqrt{|E_1|},\sqrt{|E_2|}, \frac{\sqrt{2}}{2}\sqrt{|V_1|\Delta_2},\frac{\sqrt{2}}{2}\sqrt{|V_2|\Delta_1}\right)$;}}
\end{proposition}
\begin{proof}
\mar{Using the definition of the approximation ratio guaranteed by an algorithm for a maximization problem, any upper bound of a guaranteed approximation ratio is still a guaranteed approximation ratio. Using Proposition~\ref{prop:reduction_alpha}-(iii), the optimal value $\beta(I)$ of the instance $I=\prec G_1,G_2,S \succ$ of the constrained alignment problem  equals the independence number $\alpha(\mathcal{C})$ of the related conflict graph. By Corollary~\ref{cor:alpha}, we deduce  
$$\beta(I)\leq \min \left(|E_1|,|E_2|, \frac{1}{2}|V_1|\Delta_2,\frac{1}{2}|V_2|\Delta_1\right).$$ 
Since the function $\sqrt{\cdot}$ is increasing, we conclude the proof using Proposition~\ref{pro:approxsqrt}.}
\end{proof}

\mar{Proposition~\ref{cor:approxD1} states the ratio $O(\Delta_1\log \log(\Delta_1)/\log(\Delta_1))$ in the case $m_2=1$ and $m_1$ is constant.
When $|E_1| \mar{\in} o(|\Delta_1|^2)$ or $|E_2|\mar{\in} o(|\Delta_1|^2)$, the ratio obtained in Proposition~\ref{cor:approxalpha2}-(i) can be better than the ratios we achieved  as functions of the maximum degree. In addition, Proposition~\ref{cor:approxalpha2}-(i) does not require any assumption about $m_1$. }

Given the known results for the maximum independent set, a natural question is whether the constrained alignment problem is $O(|V_1|/\log^2(|V_1|))$-approximable or even whether any approximation in $o(|V_1|\log\log(|V_1|)/\log(|V_1|)$ can be guaranteed. We give a first answer to this question in Theorem~\ref{th:approxlog} below.
The ratio $O(\sqrt{|E_1|})$ gives also a first answer for some classes of graphs satisfying $|E_1|\mar{\in} o(|V_1|^2)$ (but $\Delta_1$ still large). In particular, if $G_1$ is acyclic, we have $|E_1|\leq |V_1|$ and consequently:

\begin{corollary}\label{cor: approxacyclic}
Instances of the constrained alignment problem satisfying $m_2=1$ and $G_1$ acyclic can be approximated within the ratio $O(\sqrt{|V_1|})$.
\end{corollary}

Let now $I=\prec G_1,G_2,S \succ$ be an instance of the constrained alignment problem with conflict graph $\mathcal{C}$ and $m_2=1$; suppose we are given a subset $F\subset V_1$ and a maximal matching $M$ of $S[F\cup V_2]$, the subgraph of $S$ corresponding to similarity edges incident to $F$. We denote by $V_{\mathcal{C},F,M}$ the set of $c_4$s in $V_{\mathcal{C}}$  including at least one vertex of $F$ and no similarity edge $uv$ with $u\in F, v\in V_2, uv\notin M$; in other words, these $c_4$s include vertices in $F$ but only with similarity edges in $M$. Then, considering the subgraph $\mathcal{C}[V_{\mathcal{C},F,M}]$ of $\mathcal{C}$ induced by these $c_4$s, we have: 

\begin{lemma}\label{lem:P_4-new}
For any induced $P_3$, \mar{$x_1x_2x_3$}, in $\mathcal{C}[V_{\mathcal{C},F,M}]$, $x_1$ and $x_3$ have the same neighborhood  in $\mathcal{C}[V_{\mathcal{C},F,M}]$. In particular $\mathcal{C}[V_{\mathcal{C},F,M}]$ is $P_4$-free.
\end{lemma}

\begin{proof}
 Since $m_2=1$ and by definition of $V_{\mathcal{C},F,M}$, for every two conflicting $c_4$s in $V_{\mathcal{C},F,M}$, there must be   a vertex $u\in V_1\setminus  F$ and two disjoint vertices $v,v'\in V_2$ such that $uv$ is an edge of the former and $uv'$ an edge of the latter; moreover the other similarity edges of these $c_4$s are in $M$. Suppose we are given an induce $P_3$, $x_1x_2x_3$, in 
 $\mathcal{C}[V_{\mathcal{C},F,M}]$. There are such vertices $u,v,v'$, where $\gamma(x_1)$ and $\gamma(x_3)$ both include the edge $uv$ while $\gamma(x_2)$ includes $uv'$. Moreover, every $c_4$ in $V_{\mathcal{C},F,M}$  that conflicts  \mar{with} $\gamma(x_3)$ (resp. $\gamma(x_1)$) must include a similarity  edge $uw, w\neq v$ and thus it conflicts \mar{with} $\gamma(x_1)$ (resp. $\gamma(x_3)$), which concludes the proof.\end{proof}

We deduce the following theorem that gives a first step towards non trivial $o(|V_1|)$ approximation ratios. It corresponds to a sequence of approximation algorithms parametrized by $K$, called {\em approximation chain} in~\citet{demange-chain}.

\begin{theorem}\label{th:approxlog}
Consider instances of the constrained alignment problem satisfying $m_2=1$ and $m_1$ constant and let $K$ be a positive constant. One can find in polynomial time a legal alignment guaranteeing the approximation ratio of $\left\lceil \frac{|V_1|}{K\log(|V_1|)}\right\rceil$.
\end{theorem}

\begin{proof}
Consider an instance $I=\prec G_1,G_2,S \succ$ verifying the assumptions and denote by $\mathcal{C}$ the related conflict graph. We recall that $|V_1|\geq 2$. Denote by $\beta(I)=\alpha(\mathcal{C})$ the optimal value for the instance $I$.
Let $S^*$ be a maximum independent set of $\mathcal{C}$, $|S^*|=\alpha(\mathcal{C})$. 
Our strategy is to subdivide the vertex set of the conflict graph, $V_{\mathcal{C}}$, into $O\left(\frac{|V_1|}{\log(|V_1|)}\right)$ subsets such that the maximum independent set can be solved in polynomial time on the subgraph induced by each part. This subdivision is not necessarily a partition.

Fix a constant $K$ and partition vertices of $V_1$
 into $\mar{B_K}=\left\lceil \frac{|V_1|}{K\log(|V_1|)}\right\rceil$ sets of vertices $F_{j}, j=1, \ldots B_K$  with $|F_{j}|\leq  K\log(|V_1|)$. For each of them we denote by  $U_{j}$ the set of all $c_4$s in $V_{\mathcal{C}}$ including at least one vertex of  $F_{j}$ and by $W_{j}$ the graph $W_{j}=\mathcal{C}[U_{j}]$. Note that: 
 
 \begin{equation}\label{eq:union}
 \Union\limits_{j=1, \ldots, B_K} U_j= V_{\mathcal{C}}
 \end{equation}

We claim that there is a polynomial-time algorithm that computes, for every $j=1, \ldots, B_K$, a maximum independent set of $W_j$.  
Note first that the similarity edges involved in $c_4$s contributing to any independent set of $W_j$ form a matching of the graph $S[F_j\cup V_2]$ and consequently, is part of a maximal matching of this graph. Denoting by $\mathcal{M}_j$ the set of maximal matchings of $S[F_j\cup V_2]$, we deduce:

\begin{equation}\label{eq:maxstable}
 \alpha(W_j)=\max\limits_{M\in \mathcal{M}_j}\alpha(\mathcal{C}[V_{\mathcal{C},F_j,M}])
 \end{equation}

Lemma~\ref{lem:P_4-new} ensures that, for any fixed maximal matching $M\in \mathcal{M}_j$, $\mathcal{C}[V_{\mathcal{C},F_j,M}]$ is $P_4$-free. In this case a maximum independent set can be computed in polynomial (linear) time~(\citet{golumbicbook}). The related complexity is $O(|V_{\mathcal{C},F_j,M}|)\leq O(m_1|F_j||V_1|)$ since $c_4$s in $V_{\mathcal{C},F_j,M}$ include at least one edge of $M$ and $|M|\leq |F_j|$.  But $m_1$ is a fixed constant and $|F_{j}|\leq  K\log(|V_1|)$. Thus, we can exhaustively list all maximal matchings of $S[F_j\cup V_2]$ in $O\left( m_1^{K\log(|V_1|)} \right)=O\left( |V_1|^{K\log(m_1)} \right)$, a polynomial function.  

Our algorithm runs as follows: 
\begin{quote}
\begin{em}
For all $j=1, \ldots, B_K$ and all maximal matching $M$ of $S[F_j\cup V_2]$, compute  $\mathcal{C}[V_{\mathcal{C},F_j,M}]$ and a maximum independent set  \-- keep the best such solution. 
\end{em}
\end{quote}
\noindent
Computing each $\mathcal{C}[V_{\mathcal{C},F_j,M}]$ and a maximum independent set takes, for bounded $m_1$, $O\left(|V_1|\log(|V_1|)\right)$;  the whole complexity is then
$O\left(\log(|V_1|) |V_1|^{1+K\log(m_1)}\right)$, a polynomial function.
 
To complete the proof we need to justify it guarantees the required ratio. Equation~(\ref{eq:union}) ensures that the value $\lambda(I)$ of the computed solution satisfies:
$$
\lambda(I)=\max\limits_{j=1, \ldots, B_K}\alpha(W_j)\geq \max\limits_{j=1, \ldots, B_K}|S^*\cap U_j|\geq \frac{\beta(I)}{B_K}
$$
which shows that the related approximation ratio is $B_k=\left\lceil \frac{|V_1|}{K\log(|V_1|)}\right\rceil$. 
\end{proof}

\subsubsection{Cliques and Claws}\label{subsub:cc}

Next we present results regarding the existence of cliques as subgraphs of conflict graphs
for any $m_1$.
Assume that there is a clique $K_t$, $t\geq 1$, in $\mathcal{C}$ and let a corresponding $c_4$ associated with a vertex $x$ from this $K_t$ be $\gamma(x)=abcd$.
We partition all the corresponding $c_4$s in $K_t$ into three disjoint reference sets with respect to 
$\gamma(x)$.
Let $S_1, S_2$ consist of all the $c_4$s respectively conflicting $\gamma(x)$ with a  $Conf_{1a}$ and $Conf_{1b}$ configuration. 
Let $S_3$ be the set of all $c_4$s 
with other kinds of conflicts ( $Conf_{2}$, $Conf_{3a}$ or $Conf_{3b}$) with $\gamma(x)$ and $\gamma(x)$ itself.  

\begin{lemma}
\label{lem_edgesharing}
Given an instance $\prec G_1,G_2,S \succ$ with conflict graph $\mathcal{C}$ and the reference sets defined as above, then any pair of $c_4$s from different reference sets do not share a similarity edge.
\end{lemma}
\begin{proof}
Note that since the pair of $c_4$s 
correspond, in $\mathcal{C}$, to different vertices of the same clique $K_t$, they should conflict
by sharing at least one vertex from $G_1$. 
We consider two cases. For the first case assume one of the $c_4$s is in $S_1$ or $S_2$, and the other is in $S_3$.
Without loss of generality assume the former $c_4$ is in $S_1$ including vertices $s$ and $a$ from $G_1$, where $s\neq b$.
Since the latter $c_4$ from $S_3$ includes both $a, b$ from $G_1$, the pair of $c_4$s can only share the vertex $a$ from 
$G_1$ giving rise to a $Conf_{1a}$ or a $Conf_{1b}$ conflict between them. 
For the second case assume one of the $c_4$s is in $S_1$ and 
the other is in $S_2$. In this case the former must have a $Conf_{1a}$ conflict whereas
the latter must have a $Conf_{1b}$ conflict with the reference $\gamma(x)=abcd$. Since $a\neq b$ 
the $c_4$s from $S_1$ and $S_2$ can only share one vertex from $G_1$, thus giving rise to a $Conf_{1a}$ or a $Conf_{1b}$ conflict between the pair. 
In both cases we show that both $c_4$s are in $Conf_{1a}$ or $Conf_{1b}$ conflict with each other. 
The fact that any pair of $c_4$s with a $Conf_{1a}$ or a $Conf_{1b}$ conflict do not share a similarity edge
completes the proof.
\end{proof}

\begin{theorem}
\label{lem_clique}
Given an instance $\prec G_1,G_2,S \succ$ with conflict graph $\mathcal{C}$ and $m_2=1$, the maximum size of any clique in $\mathcal{C}$ is 
${m_1}^2$, or equivalently $\mathcal{C}$ is $K_{1+{m_1}^2}$-free.
\end{theorem}

\begin{proof} 
We consider two cases. 

 \emph{ Case-1:} We first handle the case where at least one of $S_1, S_2$ is empty. 
Assume without loss of generality $S_1$ is empty. Let $p$ be the number of similarity edges
incident to $b$ in the $c_4$s of $S_3$. Since each pair of similarity edges, one incident to $a$ and one
incident to $b$, gives rise to at most 
one $c_4$, the number of $c_4$s in $S_3$ is at most $m_1p$. 
By Lemma~\ref{lem_edgesharing}, $c_4$s in $S_3$
cannot share an edge from $S$ with the $c_4$s in $S_2$. This implies that the 
number of similarity  edges incident to $b$ in the $c_4$s of $S_2$ is at most $m_1-p$. 
Let $bc'$ be such an edge and let $S_{bc'}$ denote the set of $c_4$s in $S_2$ sharing $bc'$. Since any pair of $c_4$s from $S_{bc'}$ share a similarity edge, they 
must be in a $Conf_{3a}$ or  $Conf_{3b}$ conflict with each other
and thus must share one more vertex from $G_1$
in addition to the vertex $b$. This implies that $|S_{bc'}|\leq m_1$ which further implies a total of at most $(m_1-p)m_1$ $c_4$s in $S_2$. 
The clique consisting of $c_4$s from $S_2, S_3$ has at most ${m_1}^2$ vertices. 

 \emph{ Case-2:} Now we handle the case where $S_1$ and $S_2$ are both not empty. 
It must be the case that all $c_4$s in $S_1\cup S_2$ must share a vertex $e$ from $G_1$
such that $e\neq a$, $e\neq b$. This is due to the fact that any pair of $c_4$s, one from $S_1$ the other from $S_2$,
can only have a $Conf_{1a}$ or $Conf_{1b}$ conflict and the shared node in 
\mar{this} conflict cannot be neither $a$ nor $b$. 
Let $p,q$ be the number of edges from $S$ incident respectively to $a$ and $b$ in the $c_4$s of $S_3$. 

The number of $c_4$s in $S_3$ is at most $pq$. 
By Lemma~\ref{lem_edgesharing}, the 
number of similarity edges edges incident to $a$ in the $c_4$s of $S_1$ are at most $m_1-p$ and 
the number of similarity edges incident to $b$ in the $c_4$s of $S_2$ are at most $m_1-q$.
Let $r$ be the number of 
similarity edges incident to $e$ in the $c_4$s of $S_1$. Again by Lemma~\ref{lem_edgesharing},
the number of similarity incident to $e$ in the $c_4$s of $S_2$ are at most $m_1-r$.
This implies that the maximum number of $c_4$s in $S_1$ and $S_2$ are respectively $(m_1-p)r$
and $(m_1-q)(m_1-r)$. The size of the clique consisting of 
$c_4$s from all three reference sets 
is at most $pq+(m_1-p)r+(m_1-q)(m_1-r)$, where $1\leq p,q,r\leq m_1$.
Without loss of generality let $p\leq q$. Then we have
$pq+(m_1-p)r+(m_1-q)(m_1-r)\leq pq+(m_1-p)m_1 \leq {m_1}^2$.
\end{proof}

\begin{figure}[t]	   
\begin{center}	   
\hspace*{-.3cm}
\includegraphics[width=12.5cm]{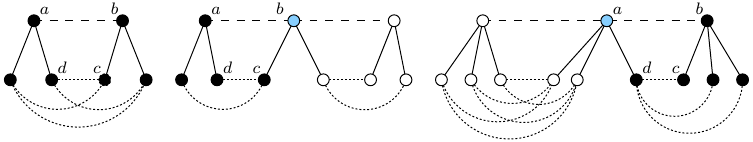} 
\caption{\sf Sample $\mathcal{C}_U$s giving rise to $K_{{m_1}^2}$s in their
respective conflict graphs.
The reference $c_4$ is $\gamma(x)=abcd$. The first two  show sample constructions for $m_1=2$ and the last for $m_1=3$.
The employed reference sets as
described in the proof of Theorem~\ref{lem_clique} are as follows:  \emph{ (Left)} All $c_4$s are in $S_3$, 
 \emph{ (Middle)} $c_4$s in $S_3$ are those induced by black vertices and $b$, $c_4$s in $S_2$ are those induced by 
white vertices and $b$, 
 \emph{ (Right)} $c_4$s in $S_3$ are those induced by black vertices and $a$, $c_4$s in $S_1$ are those induced by 
white vertices and $a$. 
} 
\label{cliquesample}	   
\end{center}	   
\end{figure}

We note that $K_{{m_1}^2}$ is possible in a conflict graph $\mathcal{C}$ 
for any positive integer $m_1$. Indeed  \emph{ Case-1} of the above proof provides
an actual construction method; see Figure~\ref{cliquesample}.

Note that under the setting of $m_2=1$, the size of $V_\mathcal{C}$ is bounded by $|E_2|$ (Lemma~\ref{vertexlemma}).
It is known that 
the maximum independent set problem is fixed-parameter tractable, parameterized by the size of the 
output, in the class of $K_r$-free 
graphs for constant integer $r$~(\citet{DBLP:conf/swat/RamanS06,DabrowskiLMR12}). Combining this result
with Theorem~\ref{lem_clique}, leads to the following result:

\begin{proposition} 
The constrained alignment problem is fixed-parameter tractable when $m_1$ is any fixed 
positive integer constant and 
$m_2=1$. 
\label{cor14}
\end{proposition} 

Note that the analogous result in~\citet{Fertin200990} 
is more restrictive since it applies only to the bounded degree
graphs.


We conclude this part by considering induced claws in conflict graphs. A  \emph{ d-claw} is an induced subgraph of an undirected graph, 
that consists of an independent set of $d$ vertices, called $talons$, and 
the $center$ vertex that is adjacent to all vertices in this set. 
Let $\Delta_{min}=min(\Delta_1, \Delta_2)$.

\begin{theorem}
\label{claw}
Given an instance $\prec G_1,G_2,S \succ$ with conflict graph $\mathcal{C}$ and $m_2=1$, then\\ $\mathcal{C}$ is $(2\Delta_{min}+2)$-claw-free.
\end{theorem}

\begin{proof}
Let $abcd$ be the corresponding $c_4$ associated with the center vertex of a claw. 
Let $abkl$ be the $c_4$ corresponding to a talon that has a $Conf_{2}$, $Conf_{3a}$ or $Conf_{3b}$  conflict with $abcd$. Since 
any other $c_4$ corresponding to a talon with a $Conf_{2}$, $Conf_{3a}$ or $Conf_{3b}$ conflict with $abcd$ would also 
have to share the vertices $a,b$, by Fact~\ref{twonodes}, 
it would conflict with $abkl$, which is not possible. Thus,
the total number of talons the $c_4$s of which create a $Conf_{2}$, $Conf_{3a}$ or $Conf_{3b}$ conflict with $abcd$
is at most 1. With regards to the number of talons corresponding, in $\mathcal{C}_U$, to  a  $Conf_{1a}$ or $Conf_{1b}$ 
conflict with $abcd$, we first count the maximum number  of possible $Conf_{1a}$ conflicts. 
Let $apqr$ be the $c_4$ of a talon with a $Conf_{1a}$  conflict with $abcd$. 
Any talon the $c_4$ of which conflicts \mar{with} $abcd$ \mar{through} a $Conf_{1a}$ conflicting configuration  must share the 
edge $ar$, since otherwise it would conflict with $apqr$. 
Any $G_1$ edge incident to vertex $a$ can belong only
to a single $c_4$ since otherwise by Fact~\ref{twonodes}
there would be a conflict between a pair of $c_4$s 
corresponding to talons. In addition, since $m_2=1$,
every $G_2$ edge can belong only to a single $c_4$. Thus
the number of talons inducing in $\mathcal{C}$ $Conf_{1a}$ conflicts is bounded by
$\Delta_{min}$. The same holds for $Conf_{1b}$ conflicts
giving rise to at most $(2\Delta_{min}+1)$ talons that 
are independent.
\end{proof}

The above theorem in conjunction with the result of~\citet{Berman00} which states 
that a $d/2$ approximation for maximum independent sets can be found 
in polynomial-time for $d$-claw free graphs gives rise to 
a polynomial-time approximation for the constrained alignment problem. 

\begin{proposition}  
If $m_2=1$, the constrained alignment problem can be $(\Delta_{min}+1)$-approximated in polynomial time.
\label{cor16}
\end{proposition}

\mar{This results improves (by at least a factor $5/6$) the approximation ratio of $2\lceil3\Delta_1/5\rceil$
for even $\Delta_1$ and $2\lceil(3\Delta_1+2)/5\rceil$ for odd 
$\Delta_1$ proposed in~(\citet{Fertin200990}). As already mentioned, the $o(\Delta_1)$-approximation stated in Proposition~\ref{cor:approxD1} already improved it. If $\Delta_2\in O(\Delta_1)$, then the ratio in Proposition~\ref{cor:approxD1} is better but if $\Delta_2\in o(\Delta_1)$, then the ratio established in  Proposition~\ref{cor16} can be better than the one in Proposition~\ref{cor:approxD1}.}

\mar{We conclude Subsection~\ref{subsec:struct-approx} by emphasizing that some of our structural results lead to a new hardness result for the maximum independent set problem.  Indeed, the combination of Lemma~\ref{deglemma}, Theorem~\ref{lemW8}, Theorem~\ref{Th:F8} and Theorem~\ref{Th:F8} states that, for any instance of the constrained alignment problem with $m_1=2, m_2=1$ and $G_1,G_2$ are of bounded degree, the related conflict graph is 
$\left(W_t\ (t\geq 5),F_6,K_5\right)$-free and of bounded degree.}

\mar{On the other hand,  the constrained alignment problem is shown to be APX-complete  even for the case where
$m_2=1, m_1=2$, both $G_1, G_2$ are bipartite and of  bounded degree (\citet{Fertin200990}).  As a consequence, we derive the following new hardness result for the maximum independent set problem:}

\begin{proposition}
The maximum independent set problem is APX-complete in the class of bounded degree, $\left(W_t\ (t\geq 5),F_6,K_5\right)$-free graphs. 
\end{proposition}

\subsection{Acyclic $G_1$ and $m_2=1$}\label{subsec:acyclic}

We conclude by investigating the case where  $G_1$ is acyclic and  $m_2=1$ for which the constrained alignment
problem is shown to be polynomial-time solvable in~\citet{AbakaBE13} without a precise complexity analysis. We refine this previous analysis by showing that in this case the conflict graph has a very particular structure. More precisely it is {\em weakly triangulated} ($C_t$-free and $\overline{C_t}$-free, for $t\geq 5$). Weakly triangulated graphs are known to be perfect~(\citet{hayward}) and moreover the maximum independent set problem can be solved in $O(|V||E|)$ in a graph $G=(V,E)$~(\citet{hayward2}). It allows us to deduce a new polynomial-time algorithm for this case with its complexity analysis. This illustrates again how the structure of the conflict graph can be used to achieve algorithmic results. 

We need two technical lemmas; remind that, given an instance $\prec G_1,G_2,S \succ$ the graph  $\mathcal{C}_U$ is defined in Subsection~\ref{subsec:conflict}.  
\begin{lemma}
\label{edgesharing}
Given an instance $\prec G_1,G_2,S \succ$ with conflict graph $\mathcal{C}$ where $G_1$ is acyclic.
Suppose a 
$P_k$, denoted by $p$, is an induced subgraph of the conflict graph $\mathcal{C}$. For $k\geq 4$, the 
$c_4$s of $\mathcal{C}_U$ corresponding to the 
end vertices of $p$ neither share a vertex nor an edge in $\mathcal{C}_U$. 
\end{lemma}

\begin{proof}
Suppose first that $p_k$ is an induced  $P_k$, $k\geq 4$ in the conflict graph and consider the two $c_4$s in $\mathcal{C}_U$ 
\mar{associated with} the extremities of $p_k$. 
 They can neither
share an edge from  $G_2$ nor a vertex from $G_2$ without sharing a similarity edge incident to it ($m_2=1$).
They also cannot share an edge from $G_1$ nor a vertex from $G_1$ without sharing a similarity edge incident to it since otherwise they would conflict.
Thus we simply need to show that they do not share 
a similarity edge.

The proof is 
by \mar{strong} induction on $k$. For the base case  $k=4$, suppose there is a $P_4$ 
\mar{$x_1x_2x_3x_4$} in the conflict graph and that 
the $c_4$s $\gamma(\mar{x_1})$ and $\gamma(\mar{x_4})$ share a similarity edge. Let $\gamma(x_1)=abcd$  and 
let $\gamma(x_4)=befc$  with the edge $bc\in S$ in common. There are two cases for 
$\gamma(x_2)$. Since it does not conflict with $\gamma(x_4)$, it  must either be of the form $gahi$, where $h,i\notin \{d,c,f\}$ ($Conf_{1a}$ conflict with $\gamma(x_1)$) or of the form $abch$ where $h\notin \{d,c,f\}$ ($Conf_{3a}$ conflict with $\gamma(x_1)$). 
Now considering $\gamma(x_3)$, to create a conflict with $\gamma(x_4)$, 
one edge of $\gamma(x_3)$ must be $ej$ where $j\notin \{d,c,f,h,i\}$. Placing the other edge of $\gamma(x_3)$ from 
\mar{$E_S$} such that 
it creates a conflict with $\gamma(x_2)$ is now impossible, since it either gives rise to a cycle in $G_1$ (cycle \mar{$abe$} or \mar{$abeg$}, $g\notin \{a,b,e\}$)  or creates a conflict with $\gamma(x_1)$.

For the inductive part, assume that the lemma holds for all $k'$ where $4\leq k'< k$. Consider the $c_4$s of $\mathcal{C}_U$ corresponding to the vertices of a $P_{k}$, 
\mar{$x_1\cdot x_{k-1}x_k$} in the conflict graph. 
Let $\gamma(x_1)=abcd$ and $\gamma(x_{k-1})=efgh$. By the inductive hypothesis, these two $c_4$s are disjoint. Consider in 
$\mathcal{C}_U$ the subset $H$ of edges from $E_1$ that belong to the $c_4$s  associated with vertices in the  $P_{k-1}$, 
$x_1\cdots x_{k-1}$. $H$ contains in particular $ab$ and $ef$. Edges in $H$ form a connected subgraph  of $G_1$ and without loss of generality we assume that the shortest path  between $b$ and $e$  contains neither $a$ nor $f$. This path  has at least one edge; let its last edge be $e'e\in G_1$ which is part of $\gamma(x_j)$ for some $j$, $1\leq j\leq (k-2)$. Let $\gamma(x_j)=e'exy$  and $\gamma(x_k)=pqrs$. If at least one of $p, q$ is on the path, say $p$, and $p\neq e'$, $p\neq e$, then $q$ must be one of $e$ or $f$, since $pqrs$ must conflict with $efgh$, which implies a cycle in $G_1$.  If $p=e'$ then $q=e$ to create a conflict with $efgh$ without creating a cycle in $G_1$.   
This implies a conflict between $pqrs$ and $e'exy$, which is impossible since  
$x_1\cdots x_k$ is an induced path. Finally, if $p=e$, $q\neq e'$, $abcd$ and $pqrs$ do not share a similarity edge, which concludes the proof.
\end{proof}

The subgraph of $\mathcal{C}_U$ that corresponds to an induced 
$P_k$, 
\mar{$x_1\cdots x_{k-1}x_k$} in the conflict graph $\mathcal{C}$
is said to be in  \emph{ chain configuration} if each $c_4$, $\gamma(\mar{x_i})$, $i=1, \ldots k$, 
shares only a distinct $G_1$-vertex with the next $c_4$, $\gamma(\mar{x_{i+1}})$, if $i<k$ and one with the previous one, $\gamma(\mar{x_{i-1}})$, if $i>1$
 and does not share any $G_1$- or $G_2$-vertices with any other of these $c_4$s; see Figure~\ref{chain} for a sample chain configuration. 
Note that a chain configuration imposes a certain order of the involved $c_4$s in $\mathcal{C}$. 

\begin{figure}[t]	   
\begin{center}	   
\includegraphics[width=8cm]{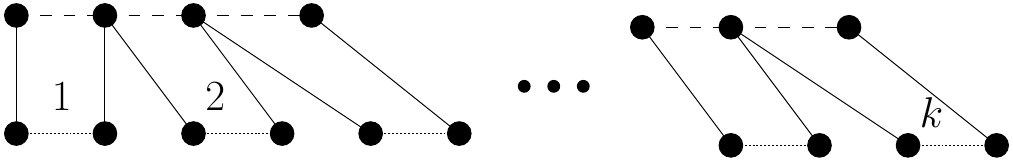} 
\caption{\sf Chain configuration of a $k$-path in $\mathcal{C}_U$. } 
\label{chain}	   
\end{center}	   
\end{figure}

\begin{lemma}
\label{threeinchain}
\mar{Let $\prec G_1,G_2,S \succ$ be an instance of the constrained alignment problem with conflict graph $\mathcal{C}$ and acyclic $G_1$. Let $x_1,x_2,x_3$ be three vertices of $\mathcal{C}$ such that $\gamma(x_1)$ and $\gamma(x_3)$ do not share a vertex nor an edge in $\mathcal{C}_U$ and $\gamma(x_2)$ conflicts  with both $\gamma(x_1)$ and $\gamma(x_3)$. Then $\gamma(x_1), \gamma(x_2)$ and $\gamma(x_3)$ must be in chain configuration where $\gamma(x_2)$ is in the middle in any left to right order.}  
\end{lemma}
\begin{proof}
If 
the conflict configuration of $\gamma(x_1)$ and $\gamma(x_2)$ were of $Conf_{2}, Conf_{3a}$ or $Conf_{3b}$, then $\gamma(x_3)$ could conflict with 
 $\gamma(x_2)$ only if it 
shared a vertex in $\mathcal{C}_U$ (more specifically a vertex from $G_1$, since $m_2=1$) with $\gamma(x_1)$, which is not possible. 
It follows that the only possible conflict configuration for 
 $\gamma(x_1)$ and $\gamma(x_2)$ 
is $ Conf_{1a}$ or $Conf_{1b}$. Applying the same reasoning to the conflict between $\gamma(x_2)$ and $\gamma(x_3)$, it follows that 
all three must be in chain configuration, where $\gamma(x_2)$ is in the middle of the chain in any left to right order. 
\end{proof}

We are now ready to prove the main result of this subsection.

\begin{theorem}
\label{th:weaklytriangl}
Given an instance $\prec G_1,G_2,S \succ$ with conflict graph $\mathcal{C}$ such that $G_1$ is acyclic and $m_2=1$ then $\mathcal{C}$ is weakly triangulated. 
\end{theorem}
\begin{proof} 
Assume first for the sake of contradiction that $C_k$  is an induced subgraph of a conflict graph for $k\geq 5$.
\mar{The cycle $C_k=x_1\cdots x_{k-1}x_kx_1$ can be divided into $(k-2)$ 
$P_3$s: $x_1x_2x_3$, $x_2x_3x_4$, $\ldots$, $x_{k-2}x_{k-1}x_k$. 
We show that for each $P_3$s, $x_ix_{i+1}x_{i+2}$, $1\leq i\leq k-2$, $\gamma(x_i), \gamma(x_{i+1}$ and $\gamma(x_{i+2})$ must be in chain configuration in $\mathcal{C}_U$. There exists indeed a $(k-1)$-path starting at vertex $x_i$ and ending 
at vertex $(x_{i+2})$ as an induced subgraph of $C_k$, thus of $\mathcal{C}$ as well. Since $k\geq 5$, by Lemma~\ref{edgesharing}, 
the $c_4$s $\gamma(x_i)$ and $\gamma(x_{i+2})$, neither   
share a vertex nor an edge in $\mathcal{C}_U$. By definition of $C_k$, they do not conflict. 
Since $\gamma(x_{i+1})$ conflicts with both $\gamma(x_i)$ and $\gamma(x_{i+2})$, by Lemma~\ref{threeinchain}, all three must be in chain 
configuration, where $\gamma(x_{i+1})$ is in the middle of the configuration in any left to right order. 
Since each of the $k-2$ triples $(x_1,x_2,x_3)$, $(x_2,x_3,x_4)$, $\ldots$, $(x_{k-2},x_{k-1},x_k)$ is in chain configuration similarly,
the $c_4$s corresponding to the whole path $x_1x_2\cdots x_{k-1}x_k$, $\gamma(x_1), \gamma(x_2) \ldots \gamma(x_{k-1}\gamma(x_k)$ 
are in chain configuration in this order. This implies there cannot be a conflict between $\gamma(x_1)$ and 
 $\gamma(x_k)$, since in the opposite case it would correspond to a cycle in graph $G_1$. This contradicts the fact vertices $x_1$ and $x_k$ are adjacent in $\mathcal{C}$}. 
 
To prove that $\overline{C_k}$ is not an induced subgraph in any conflict graph, we first note that
since $\overline{C_5}$ is isomorphic to $C_5$, $\overline{C_5}$ cannot be an induced subgraph 
of any conflict graph. For $k>5$, we 
prove it by contradiction as well. Suppose $\overline{C_k}$, with $k>5$ is an induced subgraph of $\mathcal{C}$.
\mar{Consider the path $x_{k-1}x_1x_{k-2}x_k$. This is an induced $4$-path 
in $\overline{C_k}$, thus also in $\mathcal{C}$. By Lemma~\ref{edgesharing}, $\gamma(x_k)$ and $\gamma(x_{k-1})$ 
do not share any vertex and neither                                                                         an edge in $\mathcal{C}_U$. By definition of $\overline{C_k}$ they do not conflict.
Since $\gamma(x_2)$  
conflicts with both $\gamma(x_{k-1})$ and $\gamma(x_k)$ (vertex $x_2$ is adjacent to $x_{k-1}$ and $x_k$ in $\overline{C_k}$), by Lemma~\ref{threeinchain},
$\gamma(x_{k-1}), \gamma(x_2)$, and $\gamma(x_k)$ must be in chain configuration in that order. By the same reasoning 
$\gamma(x_{k-1}), \gamma(x_3)$, and $\gamma(x_k)$ must be in chain configuration again in the same order. However this is 
only possible if  $\gamma(x_2)$ and $\gamma(x_3)$ are identical, which constitutes a contradiction.}   
\end{proof}


In~\citet{AbakaBE13}, the constrained alignment
problem is shown to be polynomial-time solvable  if $G_1$ is acyclic and $m_2=1$, using a dynamic programming approach. Theorem~\ref{th:weaklytriangl} gives an alternative proof using the $O(|V||E|)$ algorithm for maximum independent set in weakly triangulated graphs~(\citet{hayward2}). In this case, 
 Lemma~\ref{edgelemma} and Theorem~\ref{th:zagreb-bound}.\textbf{ (i)} give  $|E_\mathcal{C}|\leq \frac{1}{2}m_1^3(m_1-1)|V_1|(\Delta_1+3)$  while Lemma~\ref{vertexlemma} gives $|V_\mathcal{C}|\leq {m_1}^2 |V_1|$.
 The related complexity is $O(\Delta_1|V_1|^2)$ if $m_1$ is a fixed constant. 


\section{Concluding remarks}\label{sec:conclude}

We consider the constrained alignment of a pair of input graphs. 
We heavily investigate the combinatorial properties of a conflict graph 
which was introduced in~\citet{Fertin200990} but not studied in detail as far 
as graph theoretical properties are concerned. \mar{The constrained alignment problem appears as being closely related to the maximum independent set problem in conflict graphs.}
\mar{Known results on the maximum independent set problem associated with} several structural properties of conflict graphs 
lead to algorithmic results \mar{for the constrained alignment problem}: a polynomial-time  case, polynomial-time
approximations, and fixed-parameter tractability results.

Our contribution is twofold. First, we improve known approximation results \mar{for the constrained alignment problem} in several ways.  In terms of the maximum degrees of $G_1$ and $G_2$,  we propose the first $o(\Delta_1+\Delta_2)$-approximation using basic properties of conflict graphs. This ratio is similar to the known approximation ratios, function of the maximum degree, for the maximum independent set problem in $G_1$ and $G_2$. This is due to the fact that the maximum degree of the conflict graph is of the same order as $\Delta_1+\Delta_2$. On the contrary, the number of vertices of the conflict graph does not allow to derive interesting results from known maximum independent set approximation ratios \mar{expressed as functions of} the number of vertices. We design the first non trivial approximation result with a ratio function of $|V_1|$ for the constrained alignment problem.  The related ratio, $O(\frac{|V_1|}{\log(|V_1|)})$, is better than $O\left(\frac{|V_1|\log\log(|V_1|)}{\log(|V_1|)}\right)$ directly obtained from ratios function of the degree but it is still large compared to the $O(\frac{|V_1|}{\log^2(|V_1|)})$-approximation of the maximum independent set in $G_1$. 
\begin{quote}
A first open question raised by these results is 
to strengthen  hardness approximation results for the constrained alignment problem and in particular to investigate whether a ratio of $O(|V_1|^{1-\varepsilon})$ or even a constant approximation can be achieved in polynomial time. \mar{It is indeed well-known that such ratios cannot be achieved for the maximum independent set problem}.
\end{quote}
We also derive a ratio of $O(\sqrt{\beta(I)})$ \mar{for the constrained alignment problem with $m_2=1$}, while a similar result is not possible for the maximum independent set in general graphs unless P=NP. This kind of unusual result (Theorem~\ref{th:approx} and Proposition~\ref{pro:approxsqrt}) seems  interesting to investigate.
\begin{quote}
Studying more in detail in which extend similar $\rho(\alpha(G))$-approximation results, parametrised by the size of the optimal solution,  can be obtained for the maximum independent set problem or other problems is  an interesting line of research raised by this work.\end{quote}
 Our second contribution is about structural results on the conflict graph. \mar{After general considerations (Subsection~\ref{sec:anym1m2}) valid for any $m_1,m_2$, we focus on the case $m_2=1$ (any $m_1$) that has been considered in~\cite{Fagnot2008,Fertin200990}}. For this case, we investigate graph classes that all can be characterised by forbidden subgraphs $H$ in the neighborhood of any vertex: the case where $H$ is a large clique or a large independent set is pretty usual, it just corresponds, in the whole graph, to  exclude large cliques and/or large claws. The case where $H$ is an induced path or cycle \-- thus excluding wheels or fans \-- is less current even thought the classes of $H$-free graphs themselves have 
raised great interest in the recent years: for instance many researches deal with maximum independent set problem in graphs excluding $C_t$ or $P_t$ for some $t$. In particular, it is known that the maximum independent set  problem is  polynomial for $P_5$-free graphs~(\citet{ISP5}) and the case of larger $t$ is still unknown. For instance, if the maximum independent set problem was polynomial in $P_8$-free graphs, then combining Theorems~\ref{Th:F8} and~\ref{th:approx} would lead to a $\sqrt{\beta(I)}$-approximation for the constrained alignment problem. Note also that $P_4$-free subgraphs of $\mathcal{C}$ play a crucial role for several results in this work; it would be interesting to study whether this approach can be applied in a more general setting using $P_5$-free subgraphs instead of $P_4$-free ones. 
\begin{quote}
So far, this work motivates the study of maximum independent set in graphs excluding fans and/or wheels
and more generally in classes of graphs with forbidden subgraphs in the neighborhood of any vertex.
\end{quote}
Theorem~\ref{th:approx} gives a first step in this direction with a strategy to efficiently solve the maximum independent set in a graph $G=(V,E)$ when a good solution can be found in all subgraphs $G_v, v\in V$. 

\mar{As a first attempt to investigate properties of the conflict graph to derive efficient algorithms, the case $m_2=1$ revealed to be very rich and promising as it allows to derive interesting properties of the conflict graph, even for large values of $m_1$. Even if, as outlined in~\cite{Fagnot2008}, the underlying biological application motivates the case where both $m_1$ and $m_2$ are small, it is worth to note that reduction of the constrained alignment problem to a maximum independent set problem in the conflict graph is valid for any values of $m_1, m_2$.   As mentioned in Subsection~\ref{sub:def}, the largest possible values for $m_1, m_2$ ($m_1=|V_2|, m_2=|V_1|$)  leads to   another well studied problem, the maximum common edge subgraph problem that includes many well-known problems like the maximum clique problem. If $m_1, m_2$ are large, the size of the conflict graph increases very fast and it becomes dense. As a consequence,  this approach is likely to lead to good computational results if at least one of $m_1, m_2$ is small.} 

\begin{quote}
    \mar{The last research direction we want to outline is to investigate properties of conflict graphs for larger values of $m_2$ for, at least, some classes of graphs $G_1$ and $G_2$.}
\end{quote}

\section*{Acknowledgements}
We are grateful to the anonymous reviewers for their helpful comments and suggestions.

\bibliographystyle{abbrvnat}
\bibliography{journal}

\end{document}